\documentclass[acmsmall]{acmart}

\usepackage{booktabs}   %
\usepackage{subcaption} %

\usepackage{url}

\usepackage{amsmath}
\usepackage{semantic}
\usepackage{enumerate,paralist}
\usepackage{amssymb}
\usepackage{listing}
\usepackage{listings}
\usepackage{stmaryrd}
\usepackage{verbatim}
\usepackage{amsthm}
\usepackage{graphicx}
\usepackage{color}
\usepackage{tabularx}
\usepackage{makecell}

\usepackage[T1]{fontenc}

\lstdefinelanguage{prog}
{
	morekeywords={prob, if, then, else, fi, while, do, od, true, false, and, or, skip},
	sensitive = false
}
\newcommand{\Rset}{\mathbb{R}}
\newcommand{\Nset}{\mathbb{N}}
\newcommand{\Zset}{\mathbb{Z}}

\newcommand{\pvars}{V_\mathrm{p}}
\newcommand{\rvars}{V_{\mathrm{r}}}

\newcommand{\locs}{\mathit{L}}
\newcommand{\blocs}{\mathit{L}_{\mathrm{b}}}
\newcommand{\alocs}{\mathit{L}_{\mathrm{a}}}
\newcommand{\plocs}{\mathit{L}_{\mathrm{p}}}
\newcommand{\dlocs}{\mathit{L}_{\mathrm{d}}}
\newcommand{\transitions}{{\rightarrow}}

\newcommand{\lin}{\loc_\mathrm{in}}
\newcommand{\lout}{\loc_\mathrm{out}}

\newcommand{\val}[1]{\mbox{\sl Val}_{#1}}

\newcommand{\sampdpd}{\overline{\Upsilon}}

\newcommand{\MDPkernel}{\mathbf{P}}
\newcommand{\MDPstates}{S}
\newcommand{\MDPactions}{\mbox{\sl Act}}
\newcommand{\MDP}[1]{\mathcal{M}_{#1}}

\newcommand{\compo}{\textsf{op}}

\newcommand{\probm}{\mathbb{P}}
\newcommand{\expv}{\mathbb{E}}
\newcommand{\condexpv}[2]{{\expv}{\left({#1}{\mid}{#2}\right)}}

\newcommand{\loc}{\ell}

\newcommand{\initval}{\nu_0}
\newcommand{\initcon}{\mathbf{c}}

\newcommand{\supp}[1]{{\mathrm{supp}}{\left(#1\right)}}

\newtheorem{remark}{Remark}

\newtheorem*{theorem*}{Theorem}
\newcommand{\tm}{\textsc{Tm}}

\renewcommand{\paragraph}[1]{\smallskip\noindent\textbf{\emph{#1}}}
\renewcommand\footnotetextcopyrightpermission[1]{} %
\renewcommand\footnotetextcopyrightpermission[1]{} %
\makeatletter
\renewcommand\@formatdoi[1]{\ignorespaces}
\makeatother
\newcommand{\rd}[1]{{#1}}
\newcommand{\hongfei}[1]{{#1}}
\newcommand{\amir}[1]{{#1}}
\newcommand{\mingzhang}[1]{#1}

\begin{document}
	\fancyfoot{}\thispagestyle{empty}
	\pagestyle{plain}

\title{\hongfei{Modular} \amir{Verification for} Almost-Sure Termination \hongfei{of} Probabilistic Programs}         
\author{Mingzhang Huang}
\affiliation{
  \institution{BASICS Lab, Shanghai Jiao Tong University}            %
  \city{Shanghai}
  \country{China}                    %
}
\email{mingzhanghuang@gmail.com}          %

\author{Hongfei Fu}
\authornote{Corresponding Author}          %
\affiliation{
  \institution{Shanghai Jiao Tong University, Shanghai Key Laboratory of
Trustworthy Computing, 
East China Normal University}           %
  \city{Shanghai}
  \country{China}                   %
}
\email{fuhf@cs.sjtu.edu.cn}         %

\author{Krishnendu Chatterjee}
\affiliation{
	\institution{IST Austria (Institute of Science and Technology Austria)}           %
	\city{Klosterneuburg}
	\country{Austria}                   %
}
\email{krishnendu.chatterjee@ist.ac.at}         %

\author{Amir Kafshdar Goharshady}
\authornote{Recipient of a DOC Fellowship of the Austrian Academy of Sciences (\"{O}AW)}          %
\affiliation{
		\institution{IST Austria (Institute of Science and Technology Austria)}          %
	\city{Klosterneuburg}
	\country{Austria}                   %
}
\email{amir.goharshady@ist.ac.at}         %

\begin{abstract}
	In this work, we consider the almost-sure termination problem for probabilistic programs that asks
	whether a given probabilistic program terminates with probability~1. \amir{Scalable approaches for program analysis often rely on modularity as their theoretical basis.}
	In non-probabilistic programs, the classical variant rule (V-rule) of Floyd-Hoare logic \amir{provides} the
	foundation for \amir{modular} analysis.
	Extension of this rule to almost-sure termination of probabilistic programs is quite tricky, and
	a probabilistic variant was proposed in~\cite{HolgerPOPL}.
	While the proposed probabilistic variant cautiously addresses the key issue of integrability, we show
	that the proposed \amir{modular} rule is still not sound for almost-sure termination of probabilistic programs.

	Besides establishing unsoundness of the previous rule, our contributions are as follows:
	First, we present a sound \amir{modular} rule for almost-sure termination of probabilistic programs. Our approach is based on a novel notion of descent supermartingales.
	Second, for algorithmic approaches, we consider descent supermartingales that are linear and show that they
	can be synthesized in polynomial time.
	Finally, we present experimental results on a variety of benchmarks and several natural examples that model various types of
	nested while loops in probabilistic programs and demonstrate that our approach is able to efficiently prove their almost-sure termination property. 
\end{abstract}



\maketitle

\pagestyle{plain}

\section{Introduction}\label{sec:introduction}

\paragraph{Probabilistic programs.}
Extending classical imperative programs with randomness, i.e.~generation of random values according to probability distributions, gives rise to
probabilistic programs~\cite{gordon2014probabilistic}.
Such programs provide a flexible framework for many different applications,
ranging from the analysis of network protocols~\cite{netkat,netkat2,netkat3}, to
machine learning applications~\cite{roy2008stochastic,gordon2013model,scibior2015practical,claret2013bayesian},
and robot planning~\cite{thrun2000probabilistic,thrun2002probabilistic}.
The recent interest in probabilistic programs has led to many probabilistic programming languages
(such as Church~\cite{goodman2008church}, Anglican~\cite{anglican} and WebPPL~\cite{dippl}) and
their analysis is an active research area in formal methods and programming languages
(see~\cite{SriramCAV,pmaf,pldi18,AgrawalC018,ChatterjeeFNH16,ChatterjeeFG16,EsparzaGK12,KaminskiKMO16,kaminski2018hardness}).

\paragraph{Termination problems.}
In program analysis, the most basic liveness problem is that of {\em termination},
that given a program asks whether \amir{it} always terminates.
In presence of probabilistic behavior, there are two natural extensions of
the termination problem:
first, the almost-sure termination problem that asks whether the program terminates with probability~1; and
second, the finite termination problem that asks whether the expected termination time is finite.
While finite termination implies almost-sure termination, the converse is not true.
Both problems have been widely studied for probabilistic programs,
e.g.~\cite{kaminski2018hardness,ChatterjeeFG16,KaminskiKMO16,ChatterjeeFNH16}.

\mingzhang{
\paragraph{Importance of Almost-Sure Termination.} Almost-sure termination is the classical and most widely-studied problem that extends termination of non-probabilistic programs, and is considered as a core problem  in the programming languages community. See~\cite{HolgerPOPL,DBLP:journals/pacmpl/McIverMKK18,AgrawalC018,ChatterjeeFNH16,ChatterjeeNZ2017}.
Proving finite termination of a program is much more ideal and probably the first goal of an analyzer. Indeed, another ideal scenario would be if we could prove sure termination, i.e. that every run of the program terminates. Unfortunately, in many real-world cases, sure or finite termination are either too hard to prove or do not hold for the interesting real-world programs in question. For example, consider Recursive Markov Chains (RMCs) and Stochastic Context-free Grammars (SCFGs)~\cite{etessami2009recursive}, which are special cases of probabilistic programs and are widely used in the Natural Language Processing community~\cite{etessami2009recursive,manning1999foundations}. These programs are much simpler than general probabilistic programs, but finite termination does not hold for them, even in the real-world examples, even in the special cases such as RMCs with bounded number of entries and exits. This exact problem also appears in robot planning in AI, where many different types of random walks are known to be a.s.~terminating but not finitely terminating. In such cases, there is a natural need for proving a.s.~termination.
}

\paragraph{\amir{Modular} approaches.}
Scalable approaches for program analysis are often based on \amir{modularity} as their theoretical foundation.
For non-probabilistic programs, the classical variant rule (V-rule) of Floyd-Hoare logic~\cite{rwfloyd1967programs,DBLP:journals/acta/KatzM75}
provides the necessary foundations for \amir{modular verification}.
Such \amir{modular} methods allow decomposition of the programs into smaller parts, reasoning about the parts,
and then combining the results to deduce the desired result for the entire program.
Thus, they are the key technique in many automated methods for large programs.

\paragraph{\amir{Modular} approaches for probabilistic programs.}
The \amir{modular} approach for almost-sure termination of probabilistic programs was
considered in~\cite{HolgerPOPL}.
First, it was shown that a direct extension of the V-rule of non-probabilistic
programs is not sound for almost-sure termination of probabilistic programs, as there is a
crucial issue regarding integrability.
Then, a \amir{modular} rule, which cautiously addresses the integrability issue, was proposed
as a sound rule for almost-sure termination of probabilistic programs. We refer to this rule as the FHV-rule.

\paragraph{Our contributions.}
Our main contributions are as follows:
\begin{compactenum}
\item First, we show that the FHV-rule of~\cite{HolgerPOPL}, which is the
natural extension of the V-rule with integrability condition, is not sound for
almost-sure termination of probabilistic programs. We do this by providing a concrete counterexample in which
the FHV-rule deduces a.s.~termination whereas the program is not a.s.~terminating. We show that besides the issue of integrability, there is another
crucial issue, regarding the non-negativity requirement in ranking supermartingales, that is not addressed by~\cite{HolgerPOPL} and leads to unsoundness.

\item Second, we present a strengthened and sound \amir{modular} rule for almost-sure termination of probabilistic programs that addresses both crucial issues.
Our approach is based on a novel notion called ``descent supermartingales'' (DSMs), which is an
important technical contribution of our work.

\item \amir{Third, we provide a sound proof system for deducing a.s.~termination of programs through DSMs. This proof system can be used by theorem provers to establish a.s.~termination.}

\item Fourth, while we present our \amir{modular} approach for general DSMs,
for algorithmic approaches we focus on DSMs that are linear.
We present an efficient polynomial-time algorithm for the synthesis of linear DSMs.

\item Finally, we present an implementation of our synthesis algorithm for linear DSMs
and demonstrate that our approach \amir{can handle benchmark programs used in~\cite{pldi18}}, is applicable to probabilistic programs containing various types of nested while loops, and can efficiently prove that these programs terminate almost-surely.

\end{compactenum}

\section{Preliminaries}\label{sect:preliminaries}

Throughout the paper, we denote by $\Nset$, $\Nset_0$, $\Zset$, and $\Rset$ the sets of positive integers, non-negative integers, integers, and real numbers, respectively.
We first review several useful concepts in probability theory and then present the syntax and semantics of our probabilistic programs.

\subsection{Stochastic Processes and Martingales}\label{subsect:basics}

We provide a short review of some necessary concepts in probability theory. For a more detailed treatment, see~\cite{probabilitycambridge}.

\paragraph{Probability Distributions.} A \emph{discrete probability distribution} over a countable set $U$ is a function $p:U\rightarrow[0,1]$ such that $\sum_{u\in U}p(u)=1$.
The \emph{support} of $p$ is defined as $\supp{p}:=\{u\in U\mid p(u)>0\}$.

\paragraph{Probability Spaces.} A \emph{probability space} is a triple $(\Omega,\mathcal{F},\probm)$, where $\Omega$ is a non-empty set (called the \emph{sample space}), $\mathcal{F}$ is a \emph{$\sigma$-algebra} over $\Omega$ (i.e.~a collection of subsets of $\Omega$ that contains the empty set $\emptyset$ and is closed under complementation and countable union) and $\probm$ is a \emph{probability measure} on $\mathcal{F}$, i.e.~a function $\probm\colon \mathcal{F}\rightarrow[0,1]$ such that (i) $\probm(\Omega)=1$ and
(ii) for all set-sequences $A_1,A_2,\dots \in \mathcal{F}$ that are pairwise-disjoint
(i.e.~$A_i \cap A_j = \emptyset$ whenever $i\ne j$)
it holds that $\sum_{i=1}^{\infty}\probm(A_i)=\probm\left(\bigcup_{i=1}^{\infty} A_i\right)$.
Elements of $\mathcal{F}$ are called \emph{events}.
An event $A\in\mathcal{F}$ holds \emph{almost-surely} (a.s.) if $\probm(A)=1$.

\paragraph{Random Variables.} A \emph{random variable} $X$ from a probability space $(\Omega,\mathcal{F},\probm)$
is an $\mathcal{F}$-measurable function $X\colon \Omega \rightarrow \Rset \cup \{-\infty,+\infty\}$, i.e.~a function such that for all $d\in \Rset \cup \{-\infty,+\infty\}$, the set $\{\omega\in \Omega\mid X(\omega)<d\}$ belongs to $\mathcal{F}$.

\paragraph{Expectation.} The \emph{expected value} of a random variable $X$ from a probability space $(\Omega,\mathcal{F},\probm)$, denoted by $\expv(X)$, is defined as the Lebesgue integral of $X$ w.r.t. $\probm$, i.e.~$\expv(X):=\int X\,\mathrm{d}\probm$.
The precise definition of Lebesgue integral is somewhat technical and is
omitted  here (cf.~\cite[Chapter 5]{probabilitycambridge} for a formal definition).
If $\mbox{\sl range}~X=\{d_0,d_1,\ldots\}$ is countable, then we have
$\expv(X)=\sum_{k=0}^\infty d_k\cdot \probm(X=d_k)$.

\paragraph{Filtrations.}
A \emph{filtration} of a probability space $(\Omega,\mathcal{F},\probm)$ is an infinite sequence $\{\mathcal{F}_n \}_{n\in\Nset_0}$ of $\sigma$-algebras over $\Omega$ such that $\mathcal{F}_n \subseteq \mathcal{F}_{n+1} \subseteq\mathcal{F}$ for all $n\in\Nset_0$. Intuitively, a filtration models the information available at any given point of time.

\paragraph{Conditional Expectation.}
Let $X$ be any random variable from a probability space $(\Omega, \mathcal{F},\probm)$ such that $\expv(|X|)<\infty$.
Then, given any $\sigma$-algebra $\mathcal{G}\subseteq\mathcal{F}$, there exists a random variable (from $(\Omega, \mathcal{F},\probm)$), denoted by $\condexpv{X}{\mathcal{G}}$, such that:
\begin{compactitem}
\item[(E1)] $\condexpv{X}{\mathcal{G}}$ is $\mathcal{G}$-measurable, and
\item[(E2)] $\expv\left(\left|\condexpv{X}{\mathcal{G}}\right|\right)<\infty$, and
\item[(E3)] for all $A\in\mathcal{G}$, we have $\int_A \condexpv{X}{\mathcal{G}}\,\mathrm{d}\probm=\int_A {X}\,\mathrm{d}\probm$.
\end{compactitem}
The random variable $\condexpv{X}{\mathcal{G}}$ is called the \emph{conditional expectation} of $X$ given $\mathcal{G}$.
The random variable $\condexpv{X}{\mathcal{G}}$ is a.s. unique in the sense that if $Y$ is another random variable satisfying (E1)--(E3), then $\probm(Y=\condexpv{X}{\mathcal{G}})=1$.
We refer to ~\cite[Chapter~9]{probabilitycambridge} for details. Intuitively, $\condexpv{X}{\mathcal{G}}$ is the expectation of $X$, when assuming the information in $\mathcal{G}$.

\paragraph{Discrete-Time Stochastic Processes.}
A \emph{discrete-time stochastic process} is a sequence $\Gamma=\{X_n\}_{n\in\Nset_0}$ of random variables where $X_n$'s are all from some probability space $(\Omega,\mathcal{F},\probm)$.
The process $\Gamma$ is \emph{adapted to} a filtration $\{\mathcal{F}_n\}_{n\in\Nset_0}$ if for all $n\in\Nset_0$, $X_n$ is $\mathcal{F}_n$-measurable. Intuitively, the random variable $X_i$ models some value at the $i$-th step of the process.

\paragraph{Difference-Boundedness.}
A discrete-time stochastic process $\Gamma=\{X_n\}_{n\in\Nset_0}$ adapted to a filtration $\{\mathcal{F}_n\}_{n\in\Nset_0}$ is \emph{difference-bounded} if there exists a $c\in(0,\infty)$ such that for all $n\in\Nset_0,\ $ $|X_{n+1}-X_n|\le c$ almost-surely.

\mingzhang{
\paragraph{Martingales and Supermartingales.} A discrete-time stochastic process $\Gamma=\{X_n\}_{n\in\Nset_0}$ adapted to a filtration $\{\mathcal{F}_n\}_{n\in\Nset_0}$ is a \emph{martingale} (resp. \emph{supermartingale})
if for every $n\in\Nset_0$, $\expv(|X_n|)<\infty$ and it holds a.s. that $\condexpv{X_{n+1}}{\mathcal{F}_n} = X_n$ (resp. $\condexpv{X_{n+1}}{\mathcal{F}_n}\le X_n$).
We refer to ~\cite[Chapter~10]{probabilitycambridge} for a deeper treatment.

Intuitively, a martingale (resp. supermartingale) is a discrete-time stochastic process in which for an observer who has seen the values of $X_0, \ldots, X_n$, the expected value at the next step, i.e.~$\condexpv{X_{n+1}}{\mathcal{F}_n}$, is equal to (resp. no more than) the last observed value $X_n$. Also, note that in a martingale, the observed values for $X_0, \ldots, X_{n-1}$ do not matter given that $\condexpv{X_{n+1}}{\mathcal{F}_n} = X_n.$ In contrast, in a supermartingale, the only requirement is that $\condexpv{X_{n+1}}{\mathcal{F}_n} \leq X_n$ and hence $\condexpv{X_{n+1}}{\mathcal{F}_n}$ may depend on $X_0, \ldots, X_{n-1}.$ Also, note that $\mathcal{F}_n$ might contain more information than just the observations of $X_i$'s. 
}

\rd{
\begin{example}
	Consider an unbiased and discrete random walk, in which we start at a position $X_0$, and at each second walk one step to either left or right with equal probability. Let $X_n$ denote our position after $n$ seconds. It is easy to verify that $\expv[X_{n+1} \vert X_0, \ldots, X_n] = \frac{1}{2} (X_n - 1) + \frac{1}{2} (X_n + 1) = X_n.$ Hence, this random walk is a martingale. Note that by definition, every martingale is also a supermartingale.
	As another example, consider the classical gambler's ruin: a gambler starts with $Y_0$ dollars of money and bets continuously until he loses all of his money. If the bets are unfair, i.e.~the expected value of his money after a bet is less than its expected value before the bet, then the sequence $\{Y_n\}_{n \in \mathbb{N}_0}$ is a supermartingale. In this case, $Y_n$ is the gambler's total money after $n$ bets. On the other hand, if the bets are fair, then $\{Y_n\}_{n \in \mathbb{N}_0}$ is a martingale.
\end{example}
}

\rd{
	\begin{example}[P\'olya's Urn~\cite{mahmoud2008polya}]
		As a more interesting example, consider an urn that initially contains $R_0$ red and $B_0$ blue marbles ($R_0+B_0>0$). At each step, we take one marble from the urn, chosen uniformly at random, look at its color and then add two marbles of that color to the urn. Let $B_n, R_n$ and $M_n$ respectively be the number of red, blue and all marbles after $n$ steps. Also, let $\beta_n = \frac{B_n}{M_n}$ and $\rho_n = \frac{R_n}{M_n}$ be the proportion of marbles that are blue (resp. red) after $n$ steps. Let $\mathcal{F}_n$ model the observations until the $n$-th step. The process described above leads to the following equations:
		$$
		M_{n+1} = 1 + M_n,
		$$
		$$
		\condexpv{B_{n+1}}{\mathcal{F}_n} = \condexpv{B_{n+1}}{B_1, \ldots, B_n} = \frac{B_n}{M_n} \cdot (B_n + 1) + \frac{R_n}{M_n} \cdot B_n,
		$$
		$$
		\condexpv{R_{n+1}}{\mathcal{F}_n}=\condexpv{R_{n+1}}{B_1, \ldots, B_n} = \frac{R_n}{M_n} \cdot (R_n + 1) + \frac{B_n}{M_n} \cdot R_n.
		$$
		Note that we did not need to care about observing $R_i$'s, $M_i$'s, $\beta_i$'s or $\rho_i$'s, because they can be uniquely computed in terms of $B_i$'s. More generally, an observer can observe only $B_i$'s, or only $R_i$'s, or only $\beta_i$'s or $\rho_i$'s and can then compute the rest using this information. Based on the equations above, we have:
		$$
		\condexpv{\beta_{n+1}}{\mathcal{F}_n} = \frac{B_n}{M_n} \cdot \frac{B_{n}+1}{M_n + 1} + \frac{M_n-B_n}{M_n} \cdot \frac{B_n}{M_n + 1} = \frac{B_n}{M_n} = \beta_n, 
		$$
		$$
		\condexpv{\rho_{n+1}}{\mathcal{F}_n} = \frac{R_n}{M_n} \cdot \frac{R_{n}+1}{M_n + 1} + \frac{M_n-R_n}{M_n} \cdot \frac{R_n}{M_n + 1} = \frac{R_n}{M_n} = \rho_n. 
		$$
		Hence, both $\{\beta_n\}_{n \in \mathbb{N}_0}$ and $\{\rho_n\}_{n \in \mathbb{N}_0}$ are martingales. Informally, this means that the expected proportion of blue marbles in the next step is exactly equal to their observed proportion in the current step. This might be counter-intuitive. For example, consider a state where $99\%$ of the marbles are blue. Then, it is more likely that we will add a blue marble in the next state. However, this is mitigated by the fact that adding a blue marble changes the proportions much less dramatically than adding a red marble. 
	\end{example}
}

\subsection{Syntax}\label{subsec:syntax}

In the sequel, we fix two disjoint countable sets: the set of \emph{program variables} and the set of \emph{sampling variables}.
Informally, program variables are directly related to the control flow of a program, while sampling variables represent random inputs sampled from distributions.
We assume that
every program variable is integer-valued,
and every sampling variable is bound to a discrete probability distribution over integers.
We first define several basic notions and then present the syntax.

\paragraph{Valuations.} A \emph{valuation} over a finite set $V$ of variables is a function $\nu : V \rightarrow \mathbb{Z}$ that assigns a value to each variable. The set of all valuations over $V$ is denoted by $\val{V}$.

\paragraph{Arithmetic Expressions.} An \emph{arithmetic expression} $\mathfrak{e}$ over a finite set $V$ of variables is an expression built from the variables in $V$, integer constants, and
arithmetic operations such as addition, multiplication, exponentiation, etc.
For our theoretical results
we consider a general setting for arithmetic expressions in which the set of allowed arithmetic operations can be chosen arbitrarily.

\paragraph{Propositional Arithmetic Predicates.}
A \emph{propositional arithmetic predicate} over a finite set $V$ of variables is a propositional formula $\phi$ built from (i) atomic formulae of the form $\mathfrak{e}\Join\mathfrak{e}'$ where $\mathfrak{e},\mathfrak{e}'$ are arithmetic expressions and $\Join\in\{<,\le, >,\ge\}$, and (ii) propositional connectives such as $\vee,\wedge,\neg$.
The satisfaction relation $\models$ between a valuation $\nu$ and a propositional arithmetic predicate $\phi$ is defined through evaluation and standard semantics of propositional connectives, e.g.~(i)~$\nu\models \mathfrak{e}\Join\mathfrak{e}'$ iff $\mathfrak{e}\Join \mathfrak{e}'$ holds when the variables in $\mathfrak{e},\mathfrak{e}'$ are substituted by their values in $\nu$, (ii) $\nu\models\neg\phi$ iff $\nu\not\models\phi$ and (iii) $\nu\models\phi_1\wedge\phi_2$ (resp. $\nu\models\phi_1\vee\phi_2$) iff $\nu\models\phi_1$ and (resp. or) $\nu\models\phi_2$.

\paragraph{The Syntax.}
Our syntax is illustrated by the grammar in Figure~\ref{fig:syntax}.
Below, we explain the grammar.
\begin{compactitem}
\item \emph{Variables.} Expressions $\langle\mathit{pvar}\rangle$ (resp. $\langle\mathit{rvar}\rangle$) range over program (resp. sampling) variables.
\item \emph{Arithmetic Expressions.} Expressions $\langle\mathit{expr}\rangle$ (resp. $\langle\mathit{pexpr}\rangle$) range over arithmetic expressions over all program and sampling variables (resp. all program variables). %
\item \emph{Boolean Expressions.} Expressions $\langle\mathit{bexpr}\rangle$ range over propositional arithmetic predicates over program variables.
\item \emph{Programs.} A program can either be a single assignment statement (indicated by `$:=$'), or
`\textbf{skip}' which is the special statement that does nothing, or a
conditional branch (indicated by `\textbf{if}  $\langle\mathit{bexpr}\rangle$'),
or a non-deterministic branch (indicated by `\textbf{if}  $\mbox{$\star$}$'),
or a probabilistic branch (indicated by `\textbf{if} \textbf{prob}($p$)', where $p\in[0,1]$ is the probability of executing the \textbf{then} branch and $1-p$ that of the \textbf{else} branch),
or a while loop (indicated by the keyword `\textbf{while}'), or a sequential composition of two subprograms (indicated by semicolon).
\end{compactitem}

\paragraph{Program Counters.}
We assign a \emph{program counter} to each assignment statement, skip, if branch and while loop.
Intuitively, the counter specifies the current point in the execution of a program. We also refer to program counters as \emph{labels}.

\begin{figure}

\[\begin{array}{rrl}
\langle \mathit{prog}\rangle &::=&  \mbox{`\textbf{skip}'}\\
&& \mid\langle\mathit{pvar}\rangle \,\mbox{`$:=$'}\, \langle\mathit{expr} \rangle\\
&&  \mid\langle\mathit{prog}\rangle \, \text{`;'} \langle\mathit{prog}\rangle\\
&&  \mid\mbox{`\textbf{if}'} \, \langle\mathit{bexpr}\rangle\,\mbox{`\textbf{then}'} \, \langle \mathit{prog}\rangle \, \mbox{`\textbf{else}'} \, \langle \mathit{prog}\rangle \,\mbox{`\textbf{fi}'}\\
&& \mid \mbox{`\textbf{if}'} \, \mbox{`$\star$'}\,\mbox{`\textbf{then}'} \, \langle \mathit{prog}\rangle \, \mbox{`\textbf{else}'} \, \langle \mathit{prog}\rangle \,\mbox{`\textbf{fi}'}\\
&& \mid \mbox{`\textbf{if}'} \, \mbox{`\textbf{prob$(p)$}'}\,\mbox{`\textbf{then}'} \, \langle \mathit{prog}\rangle \, \mbox{`\textbf{else}'} \, \langle \mathit{prog}\rangle \,\mbox{`\textbf{fi}'}\\
&&\mid  \mbox{`\textbf{while}'}\, \langle\mathit{bexpr}\rangle \, \text{`\textbf{do}'} \, \langle \mathit{prog}\rangle \, \text{`\textbf{od}'}
\\
\vspace{\baselineskip}
\\
\langle\mathit{literal} \rangle &::=& \langle\mathit{pexpr} \rangle\, \mbox{`$\leq$'} \,\langle\mathit{pexpr} \rangle \mid \langle\mathit{pexpr} \rangle\, \mbox{`$\geq$'} \,\langle\mathit{pexpr} \rangle
\\
\langle \mathit{bexpr}\rangle &::=&  \langle \mathit{literal} \rangle \mid \neg \langle\mathit{bexpr}\rangle \mid \langle \mathit{bexpr} \rangle \, \mbox{`\textbf{or}'} \, \langle\mathit{bexpr}\rangle
\\
&&\mid \langle \mathit{bexpr} \rangle \, \mbox{`\textbf{and}'} \, \langle\mathit{bexpr}\rangle
\end{array}\]

\caption{The syntax of probabilistic programs}
\label{fig:syntax}
\end{figure}

\subsection{Semantics}

To specify the semantics of our probabilistic programs, we follow previous approaches, such as~\cite{SriramCAV,ChatterjeeFNH16,ChatterjeeFG16}, and use Control Flow Graphs (CFGs) and Markov Decision Processes (MDPs) (see~\cite[Chapter 10]{DBLP:books/daglib/0020348}).
Informally, a CFG describes how the program counter and valuations over program variables
change along an execution of a program.
Then, based on the CFG, one can construct an MDP as the semantical model of the probabilistic program.

\begin{definition}[Control Flow Graphs]
\label{def:CFG}
A \emph{Control Flow Graph} (CFG) is a tuple
\begin{equation}\label{eq:cfg}
\mathcal{G}=(\locs,(\pvars,\rvars),\transitions)
\end{equation}
with the following components:
\begin{compactitem}
\item $\locs$ is a finite set of \emph{labels}, which is partitioned into the set $\blocs$ of \emph{conditional branch} labels, the set $\alocs$ of \emph{assignment} labels,
       the set $\plocs$ of \emph{probabilistic} labels and the set $\dlocs$ of \emph{non-deterministic branch} labels;
\item $\pvars$ and $\rvars$ are disjoint finite sets of \emph{program} and \emph{sampling} variables,  respectively;
\item $\transitions$ is a \emph{transition relation} in which every member (called a \emph{transition}) is a tuple of the form
$(\loc,\alpha,\loc')$ for which (i) $\loc$ (resp. $\loc'$) is the \emph{source label} (resp. \emph{target label}) in $\locs$
and (ii) $\alpha$ is either a propositional arithmetic predicate if $\loc\in\blocs$, or an \emph{update function} $u:\val{\pvars}\times\val{\rvars}\rightarrow \val{\pvars}$ if $\loc\in\alocs$,
or $p\in[0,1]$ if $\loc\in\plocs$
or $\star$ if $\loc\in\dlocs$.
\end{compactitem}
We always specify an \emph{initial} label $\lin\in\locs$ representing the starting point of the program, and
a \emph{terminal} label $\lout\in\locs$
that represents termination and has no outgoing transitions.
\end{definition}

\paragraph{Intuition for CFGs.} Informally, a control flow graph specifies how the program counter and values for program variables change in a program.
We have three types of labels, namely \emph{branching},  \emph{assignment} and \emph{non-deterministic}.
The initial label $\lin$ corresponds to the initial statement of the program.
A conditional branch label corresponds to a conditional branching statement indicated by `\textbf{if} $\phi$' or `\textbf{while} $\phi$'
, and leads to the next label determined by $\phi$ without changing the valuation.
An assignment label corresponds to an assignment statement indicated by `$:=$' or $\mathbf{skip}$,
and leads to the next label right after the statement and an update to the value of the variable on the left hand side of `$:=$' that is specified by its right hand side. This update can be seen as a function that gives the next valuation over program variables based on the current valuation and the sampled values. The statement `\textbf{skip}' is treated as an assignment statement that does not change values.
A probabilistic branch label corresponds to a probabilistic branching statement indicated by `\textbf{if} \textbf{prob}$(p)$', and leads to the label of `\textbf{then}' (resp. `\textbf{else}') branch with probability $p$  (resp. $1-p$).
A non-deterministic branch label corresponds to non-deterministic choice statement indicated by `\textbf{if} $\star$', and has transitions to the two labels corresponding to the `\textbf{then}' and `\textbf{else}' branches.

By standard constructions, one can transform any probabilistic program into an equivalent CFG.
We refer to ~\cite{SriramCAV,ChatterjeeFNH16,ChatterjeeFG16} for details.

\begin{example}\label{example:cfg}
Consider the probabilistic program in Figure~\ref{example:counter}~(left). Its CFG is given in Figure~\ref{example:counter}~(right).
In this program, $x$, $y$ and $z$ are program variables, and $r$ is a sampling variable that observes the probability distribution $\probm(r=1)=\probm(r=-1)=0.5$.
The numbers $1$--$10$ are the program counters (labels).
In particular, $1$ is the initial label and $10$ is the terminal label.
The arcs represent transitions in the CFG.
For example, the arc from $5$ to $7$ specifies the transition from label $5$ to label $7$ with the update function $x\mapsto x+r$ that assigns to program variable $x$, the value of the expression $x+r$, obtained by adding the value of $x$ to a sampled value for the sampling variable $r$.
\end{example}

\begin{figure}
	\begin{minipage}{0.45\linewidth}
	\lstset{language=prog}
	\lstset{tabsize=4}	
	\begin{lstlisting}[mathescape,basicstyle=\small]
	1:	while $x\geq 1$ do
	2:		$z:=y$;
	3:		while $z\geq0$ do
	4:			if $x<2$ then
	5:				$x:=x+r$
				else
	6: 				skip
	  			fi;
	7:			$z:=z-1$
      		od;
	8:		$y:=4 \times y$;
	9:		$x:=x-1$
	  	od
	10:
	\end{lstlisting}
	\end{minipage}
	\begin{minipage}{0.45\linewidth}
		\includegraphics[width=0.95\linewidth,keepaspectratio]{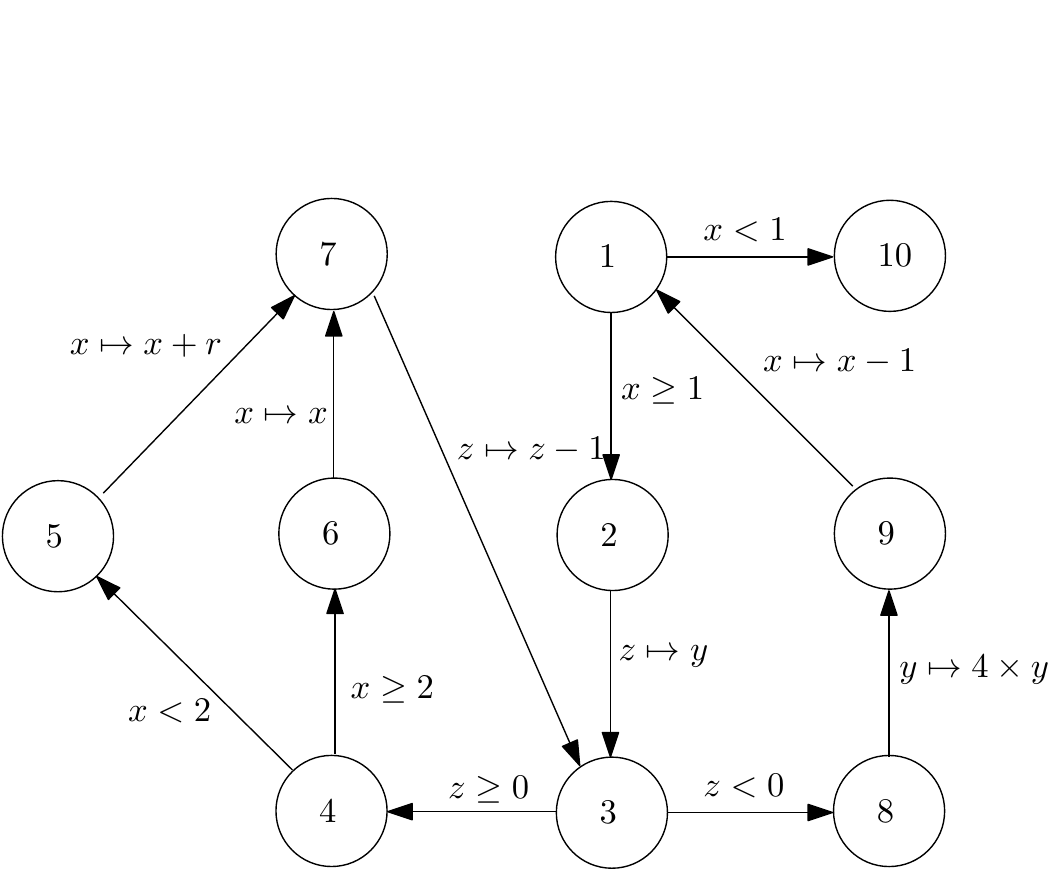}
	\end{minipage}
	\caption{A probabilistic program (left) and its CFG (right). In this program, $\probm(r=1) = \probm(r=-1) = 0.5.$}\label{example:counter} \label{fig:cfgrunning}
\end{figure}

\paragraph{The Semantics.} Based on CFGs, we define the semantics of probabilistic programs through the standard notion of Markov decision processes.
Below, we fix a probabilistic program $P$ with its CFG in form (\ref{eq:cfg}).
We define the notion of \emph{configurations} such that a configuration is a pair $(\loc,\nu)$, where $\loc$ is a label (representing the current program counter) and $\nu\in\val{\pvars}$ is a valuation (representing the current valuation for program variables).
We also fix a \emph{sampling function} $\Upsilon$ which
assigns to every sampling variable $r\in\rvars$, a discrete probability distribution over $\Zset$.
Then, the \emph{joint} discrete probability distribution $\sampdpd$ over $\val{\rvars}$ is defined as
$\sampdpd(\mu):=\prod_{r\in\rvars} \Upsilon(r)(\mu(r))$ for all valuations $\mu$ over sampling variables.

The semantics is described by a Markov decision process (MDP).
Intuitively, the MDP models the stochastic transitions, i.e.~how the current configuration jumps to the next configuration.
The state space of the MDP is the set of all configurations. %
The actions are $\tau$, $\textbf{th}$ and $\textbf{el}$ and correspond to
the absence of non-determinism, taking the \textbf{then}-branch of a non-deterministic branch label, and taking the $\textbf{else}$-branch of
a non-deterministic branch label, respectively.
The MDP transition probabilities are determined by the current configuration, the action chosen for the configuration and the statement at the current configuration.

To resolve non-determinism in MDPs, we use schedulers.
A \emph{scheduler} $\sigma$ is a function which maps every history, i.e.~all information up to the current execution point, to a probability distribution over the actions available at the current state. Informally, it resolves non-determinism
by discrete probability distributions over actions that specify the probability of taking each action.

From the MDP semantics, the behavior of a probabilistic program $P$ with its CFG in the form~(\ref{eq:cfg})
is described as follows:
Consider an arbitrary scheduler $\sigma$.
The program starts in an initial configuration $(\loc_0,\initval)$ where $\loc_0=\lin$.
Then in each step $i$ ($i\ge 0$), given the current configuration $(\loc_{i},\nu_{i})$, the next configuration $(\loc_{i+1},\nu_{i+1})$ is determined as follows:
\begin{compactenum}
\item
a valuation $\mu_i$ of the sampling variables is sampled according to the joint distribution $\sampdpd$; %
\item
if $\loc_i\in\alocs$ and $(\loc_i, u, \loc')$ is the transition in $\transitions$ with source label $\loc_i$ and update function $u$, then $(\loc_{i+1},\nu_{i+1})$ is set to be
$(\loc',u(\nu_{i},\mu_i))$;
\item
if $\loc_i\in\blocs$ and $(\loc_i, \phi, \loc'),(\loc_i, \neg\phi,\loc'')$ are the two transitions in $\transitions$ with source label $\loc_i$, then
$(\loc_{i+1},\nu_{i+1})$ is set to be either (i) $(\loc',\nu_i)$ if $\nu_i\models\phi$, or (ii) $(\loc'',\nu_i)$ if $\nu_i\models\neg\phi$;
\item
if $\loc_i\in\dlocs$ and $(\loc_i, \star, \loc')$, $(\loc_i, \star, \loc''$) are the transitions in $\transitions$ with source label $\loc_i$, then $(\loc_{i+1},\nu_{i+1})$ is set to be
$(\loc''',\nu_{i})$, where the label $\loc'''$ is chosen from $\loc',\loc''$ using the scheduler $\sigma$.
\item
if $\loc_i\in\plocs$ and $(\loc_i, p, \loc'),(\loc_i, 1-p,\loc'')$ are the two transitions in $\transitions$ with source label $\loc_i$, then
$(\loc_{i+1},\nu_{i+1})$ is set to be either (i) $(\loc',\nu_i)$ with probability $p$, or (ii) $(\loc'',\nu_i)$ with probability $1-p$;

\item
if there is no transition in $\transitions$ emitting from $\loc_i$ (i.e.~if $\loc_i=\lout$), then $(\loc_{i+1}, \nu_{i+1})$ is set to be $(\loc_i,\nu_{i})$.
\end{compactenum}
For a detailed construction of the MDP, see Appendix~\ref{app:semantics}.

\paragraph{Runs and the Probability Space.}
A \emph{run} is an infinite sequence of configurations.
Informally, a run $\{(\loc_{n},\nu_{n})\}_{n\in\Nset_0}$ specifies that the configuration at the $n$-th step
of a program execution is $(\loc_{n},\nu_{n})$, i.e.~the program counter (resp. the valuation for program variables) at the $n$-th step is $\loc_{n}$ (resp. $\nu_{n}$).
By construction, with an initial configuration $\initcon$
(as the initial state of the MDP) and a scheduler $\sigma$,
the Markov decision process for a probabilistic program induces a unique probability space over the runs (see ~\cite[Chapter 10]{DBLP:books/daglib/0020348} for details).
In the rest of the paper, we denote by $\probm^\sigma_\initcon$ the probability measure under the initial configuration $\initcon$ and the scheduler $\sigma$, and by $\expv^\sigma_\initcon(-)$ the corresponding expectation operator.

\section{Problem Statement}\label{sec:ps}
In this section, we define the \amir{modular} verification problem of almost-sure termination over probabilistic programs.
Below, we fix a probabilistic program $P$ with its CFG in the form (\ref{eq:cfg}).
We first define the notion of almost-sure termination.
Informally, the property of almost-sure termination requires that a program terminates with probability~$1$.
We follow the definitions in ~\cite{SriramCAV,HolgerPOPL,ChatterjeeFNH16}.

\begin{definition}[Almost-sure Termination]\label{def:astermination}
A run $\omega=\{(\loc_n,\nu_n)\}_{n\in\Nset_0}$ of a program $P$ is \emph{terminating} if  $\loc_n=\lout$ for some $n\in\Nset_0$.
We define the \emph{termination time} as a random variable $T$ such that for a run $\omega=\{(\loc_n,\nu_n)\}_{n\in\Nset_0}$, $T(\omega)$ is the smallest $n$ such that $\loc_n=\lout$ if such an $n$ exists (this case corresponds to program termination), and $\infty$ otherwise (this corresponds to non-termination).
The program $P$ is said to be \emph{almost-surely (a.s.) terminating} under initial configuration $\initcon$ if $\probm^\sigma_\initcon(T<\infty)=1$ for all schedulers $\sigma$.
\end{definition}

\begin{lemma}\label{lem:seqcond}
	Let the program $P$ be the sequential (resp.~branching) composition of two other programs $P_1$ and $P_2$, i.e.~$P := P_1 ; P_2$ (resp.~$P := \textbf{if}~~- ~~\textbf{then}~~P_1~~\textbf{else}~~P_2~~\textbf{fi}$), and assume that both $P_1$ and $P_2$ are a.s. terminating for any initial value. Then, $P$ is also a.s. terminating for any initial value.
\end{lemma}

\mingzhang{
\begin{proof}
	We first prove the sequential case. Let $V_p = \{x_1, x_2, \ldots, x_m\}$ be the set of program variables in $P$ and $T, T_1, T_2$ be the termination time random variables of $P$, $P_1$ and $P_2$, respectively. Define the vector $F_1$ of random variables as follows: if $\omega = \{ (\loc_j, \nu_j)\}_{j \in \mathbb{N}_0}$ is a terminating run of $P_1$ with $T_1(\omega)=n$, i.e.~if $\omega$ terminates at $(\loc_n, \nu_n)$, then $F_1(\omega) = \nu_n$. Intuitively, $(F_1(\omega))_i$ is the random variable that models the value of the $i$-th program variable at termination time of $P_1$. Then, we have:
	$$
	\probm_\initcon^\sigma(T < \infty) = \sum_{\nu \in \val{V_p}} \probm_\initcon^\sigma (T_1 < \infty ~~ \wedge ~~ F_1 = \nu) \cdot \probm_\nu^\sigma(T_2 < \infty).
	$$
	Informally, $P$ terminates if and only if $P_1$ terminates and then $P_2$, run with the initial valuation obtained from the last step of $P_1$, terminates as well. However, $P_2$ is a.s. terminating, hence $\probm_\nu^\sigma(T_2 < \infty) = 1$ for all $\nu$. Therefore,
	\begin{eqnarray*}
		\probm_\initcon^\sigma(T < \infty) & =  & \sum_{\nu \in \val{V_p}} \probm_\initcon^\sigma (T_1 < \infty ~~ \wedge ~~ F_1 = \nu) \\ & = & \probm_\initcon^\sigma(T_1 < \infty)\\ & = & 1,
	\end{eqnarray*}
	so $P$ is also a.s. terminating.
	
	For the branching case, note that if $P$ does not terminate, then at least one of $P_1$ and $P_2$ does not terminate as well. Formally, $\probm_{\initcon}^\sigma(T < \infty) \geq \probm_{\initcon}^\sigma(T_1 < \infty) \cdot \probm_{\initcon}^\sigma(T_2 < \infty) = 1.$
\end{proof}
}

\begin{remark}
 The lemma above shows that a.s. termination is closed under branching and sequential composition. Hence,
in this paper, we
\amir{focus on} the major problem of \amir{modular}
verification for a.s.~termination of nested while loops.
\end{remark}

\paragraph{\hongfei{Modular} \amir{Verification.}} We now define the problem of \hongfei{modular} \amir{verification} of a.s. termination,  \amir{following the terminology of~\cite{DBLP:conf/compos/KupfermanV97}}.
We first describe the notion of \hongfei{modularity} in general.
Consider an
operator $\compo$ (e.g.~sequential composition or loop nesting) over \amir{a set of} objects.
We say that a property $\phi$ is \hongfei{modular} under the operator $\compo$
with a \emph{side condition} $\psi$ (\hongfei{over pairs of objects}) if we have
\begin{equation}\label{def:compo}
(\psi(O,O')\wedge\phi(O)\wedge \phi(O'))\Rightarrow\phi(\compo(O,O'))
\end{equation}
holds for all objects $O,O'$.
In other words, the \hongfei{modularity} of the property $\phi$ says that
if the side condition $\psi$ and the property $\phi$ hold for $O,O'$, then the property $\phi$ also holds on the (bigger) composed object $\compo(O,O')$.
The motivation \amir{behind} \hongfei{modular} \amir{verification} is that \hongfei{it allows one to prove the property incrementally from sub-objects to the whole object}.

\paragraph{The Almost-sure Termination Property.}
In this paper, we are concerned with a.s. termination of while loops. Our aim is to prove this based on the assumption that the loop body is a.s. terminating for \emph{every} initial value. We consider ${\textsc{Tm}}(P)$ to be the target property expressing that the probabilistic program $P$ is a.s. terminating for \emph{every} initial value, and we consider the operator to be the while-loop operator $\textbf{while}$, i.e.~given a probabilistic program $P$ and a propositional arithmetic predicate $G$ (as the loop guard), the probabilistic program $\textbf{while}(G,P)$ is defined as $\textbf{while}~G~\textbf{do}~P~\textbf{od}$.
Since $P$ might itself \amir{contain} another while loop, our setting encompasses probabilistic nested loops of any depth.

We focus on the \amir{modular} verification of ${\textsc{Tm}}(-)$ under the while-loop operator and solve the problem in the following steps.
First, we establish a side condition $\psi$ so that the assertion
\begin{equation}
(\psi(G,P)\wedge\textsc{Tm}(P))\Rightarrow \textsc{Tm}(\textbf{while}(G,P))\label{pro:compana}
\end{equation}
holds for all probabilistic programs $P$ and propositional arithmetic predicates $G$\footnote{Note that we do not define or consider any assertion of the form $\textsc{Tm}(G)$, because checking the condition $G$ always takes finite time.}.
Second, based on the proposed side condition, we explore possible \hongfei{proof-rule} and algorithmic approaches.

\hongfei{
\begin{remark}
	 In modular \amir{verification}, the level of modularity depends \amir{on the complexity of} the side condition. The least amount of modularity is \amir{achieved when} the side condition \amir{is} equivalent to the target property, so that no information from sub-objects is used. \amir{On the other hand,  maximum} modularity is \amir{attained when there is no need to prove a side condition}, resulting in \amir{what is sometimes called} \emph{compositionality} of the property\amir{\footnote{Some authors use the terms ``modular'' and ``compositional'' interchangeably, while others require compositional approaches to have no side conditions. To avoid confusion, we always describe our approach as modular.}}. In our modular approach to prove a.s. termination, we \amir{consider a} side condition that lies in the middle: the side condition neither encodes the \amir{entire} a.s.~termination property, nor \amir{can it be removed}. We note that \amir{it is not possible to find} a modular approach for proving a.s.~termination \amir{for} nested while loops if we
	 forbid side conditions on the loop guard and the loop body, \amir{because} the variables in a loop guard can be updated in the inner loops.
\end{remark}
}

\section{Previous Approaches for Modular Verification of Termination}\label{sect:nontriviality}

In this section, we describe previous approaches for \amir{modular} verification of the (a.s.) termination property for (probabilistic) while loops.
We first present the variant rule from the Floyd-Hoare logic~\cite{rwfloyd1967programs, DBLP:journals/acta/KatzM75} that is sound for non-probabilistic programs.
Then, we describe the probabilistic extension proposed in~\cite{HolgerPOPL}.
\subsection{Classical Approach for Non-probabilistic Programs}\label{sect:vrulenonp}

Consider a non-probabilistic while loop
\begin{equation*}\label{eq:nloop}
P=\textbf{while } G \textbf{ do } P_1;\dots;P_n\textbf{ od}
\end{equation*}
\noindent where the programs $P_1, \dots, P_n$ may contain nested while loops and are assumed to be terminating.
The fundamental approach for \amir{modular verification} is the following classical variant rule (V-rule) from the Floyd-Hoare logic~\cite{rwfloyd1967programs, DBLP:journals/acta/KatzM75}:
$$  \inference[\text{V-RULE}]{\forall k:P_k\text{ terminates and }\{R=z\}P_k\{R\preceq z\}\\ \exists k\,\{R=z\}P_k\{R\prec z\} }{\textbf{while } G \textbf{ do } P_1;\ldots ;P_n\textbf{ od }\ \text{terminates}}$$
In the V-rule above, $R$ is an arithmetic expression over program variables that acts as a
\emph{ranking function}. The relation $\prec$ represents a well-founded relation
when restricted to the loop guard $G$, while the relation $\preceq$ is the ``non-strict'' version of $\prec$ such that
(i) $a\prec b\wedge b\preceq c\Rightarrow a\prec c$ and (ii) $a\preceq b\wedge b\prec c\Rightarrow a\prec c$.
Then, the premise of the rule says that (i)~for all $P_k$, the
value of $R$ after the execution of $P_k$ does not increase
in comparison with its initial value $z$ before the execution, and (ii)~there is some $k$ such that the execution of $P_k$ leads to a decrease in the value of $R$.
If $\{R=z\}P_k\{R\preceq z\}$ holds, then $P_k$ is said to be \emph{unaffecting} for $R$. Similarly, if $\{R=z\}P_k\{R\prec z\}$ holds, then $P_k$ is \emph{ranking} for $R$. %
Informally, the variant rule says that if all $P_k$'s are unaffecting and there is at least one $P_k$ that is ranking, then $P$ terminates.

The variant rule is sound for proving termination of non-probabilistic programs, because the value of $R$ cannot be decremented infinitely many times, given that the relation $\prec$ is well-founded when restricted to the loop guard $G$.

\subsection{A Previous Approach for Probabilistic Programs}

The approach in~\cite{HolgerPOPL} can be viewed as an extension of the abstract V-rule, which is a proof system for a.s. terminating property.
We call this abstract rule the FHV-rule:
  $$\inference[\text{FHV-RULE}]{\forall k:P_k\text{ terminates and }\{R=z\}P_k\{R\preceq z\}\\ \exists k\,\{R=z\}P_k\{R\prec z\} }{\textbf{while } G \textbf{ do } P_1;\ldots ;P_n\textbf{ od }\ \text{terminates}}$$
   Note that while the FHV-rule looks identical to the V-rule, semantics of the Hoare triple in the FHV-rule are different from that of the V-rule (see below).

The FHV-rule is a direct probabilistic extension of the V-rule through the notion of \emph{ranking supermartingales} (RSMs, see ~\cite{SriramCAV,ChatterjeeFNH16,ChatterjeeFG16}).
RSMs are discrete-time stochastic processes that satisfy the following conditions: (i)~their values are always non-negative; and (ii)~at each step of the process, the conditional expectation of the value is decreased by at least a positive constant $\epsilon$.
The decreasing and non-negative nature of RSMs ensures that with probability $1$ and in finite expected number of steps, the value of any RSM hits zero.
When embedded into programs through the notion of RSM-maps (see e.g.~\cite{SriramCAV,ChatterjeeFNH16}),
RSMs serve as a sound approach for proving termination of probabilistic programs with finite expected \hongfei{termination} time, which implies a.s. termination, too.

In~\cite{HolgerPOPL}, the $R$ in the FHV-rule is a propositionally linear expression that represents an RSM, while
$\prec$ is the well-founded relation on non-negative real numbers such that $x\prec y$ iff $x\le y-\epsilon$ for some fixed positive constant $\epsilon$ and $\preceq$ is interpreted simply as $\le$.
Unaffecting and ranking conditions are extended to the probabilistic setting through conditional expectation  (see $Dec_{\le} (-,-), Dec_{<} (-,-)$ on~\cite[Page 9]{HolgerPOPL}). Concretely, we say that (i) $P_k$ is \emph{unaffecting} if the expected value of $R$ after the execution of $P_k$ is no greater than its initial value before the execution; and (ii)
$P_k$ is \emph{ranking} if the expected value of $R$ after the execution of $P_k$ is decreased by at least $\epsilon$ compared with its initial value before the execution.
Note that in~\cite{HolgerPOPL}, $R$ is also called a \emph{compositional RSM}.

\paragraph{Crucial Issue 1: Difference-boundedness and Integrability.}
In~\cite{HolgerPOPL}, the authors accurately observed that simply extending the variant rule with expectation is not enough.
They provided a counterexample in~\cite[Section 7.2]{HolgerPOPL} that is not a.s. terminating but has a compositional RSM.
The problem is that random variables may not be integrable after the execution of a probabilistic while loop.
In order to resolve this integrability issue, they introduced the \hongfei{conditional} difference-boundedness condition (see Section~\ref{sect:preliminaries}) for conditional expectation. %
Then, using the Optional Sampling/Stopping Theorem, they proved that, under \hongfei{this}
condition, the random variables are integrable after the execution of while loops.
To ensure the \hongfei{conditional} difference-bounded condition, they established sound inference rules
(see~\cite[Table 2 and Theorem 7.6]{HolgerPOPL}).
With the integrability issue resolved,~\cite{HolgerPOPL} finally claims that
compositional ranking supermartingales provide a sound approach for proving a.s. termination of
probabilistic while loops (see~\cite[Theorem 7.7]{HolgerPOPL}).
\section{A Counterexample to the FHV-rule}

Although~\cite{HolgerPOPL} takes care of the integrability issue, we show that, unfortunately, the FHV-rule is still not sound.
We present an explicit counterexample on which the FHV-rule proves a.s. termination, while the program is actually not a.s. terminating.

\begin{example}[The Counterexample]\label{ex:counterexample}
\begin{figure}
  \includegraphics{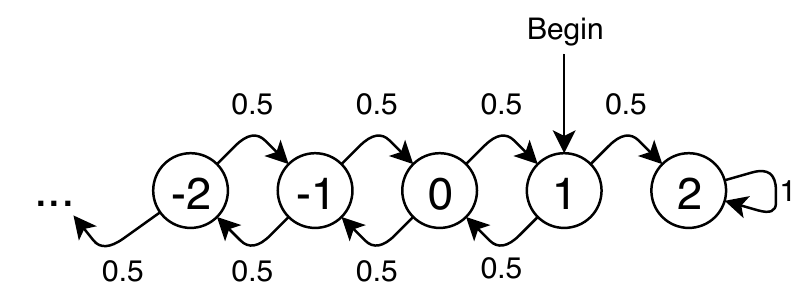}
  \caption{The random walk modeled by the inner loop of the counterexample}~\label{fig:randomwalk}
\end{figure}
\mingzhang{
Consider a $1$-dimensional symmetric random walk as in Figure~\ref{fig:randomwalk}, in which at every location there is a $\frac{1}{2}$ probability of going left and $\frac{1}{2}$ probability of going right, except that there is a barrier at location $2$, i.e.~we remain at $2$ if we reach it. Our inner loop in Figure~\ref{fig:cfgrunning} models this random walk, where the variable $x$ corresponds to the location, and the number of steps taken is $z=y$ (lines $2-3$). In our outer loop, we run this random walk many times, but each time we increase the number of steps exponentially (line $8$) and also, to avoid starting at the barrier, move the location one step to the left (line $9$), before running the walk again.

Note that the program terminates only if $x\neq 2$, i.e.~if we do \emph{not} reach the barrier after some step-bounded random walk (lines $9$ and $1$). It is well-known that a $1$-d symmetric random walk with a barrier will eventually reach the barrier with probability $1$. In our program, the probability of reaching the barrier increases as we increase the number of steps of the random walk.

Moreover, at each iteration of the outer loop we are dramatically increasing the probability of reaching the barrier by increasing the number of steps by a factor of $4$ (line $8$). Therefore, after each iteration of the outer loop, the probability of termination in the next iteration is dramatically decreased and as a result, the program as a whole does not terminate a.s. This argument is formalized in Proposition~\ref{prop:count1}.

We now argue that the FHV-rule wrongly deduces a.s.~termination for this program. Intuitively, all the FHV-rule demands to assert termination is that (i)~the inner loop terminates a.s., and~(ii) there is an integrable expression $R$ such that its expected value decreases in each iteration. Point~(i) is trivial as the inner loop takes at most $z$ steps. For point~(ii), we let $R=x$, i.e.~we use our location in the random walk as the ranking expression. It is easy to verify that the value of $x$ does not change in expectation in the random walk. So, it decreases by $1$ in each iteration of the outer loop (line $9$). Hence, the FHV-rule incorrectly concludes that the program in Figure~\ref{fig:cfgrunning} terminates a.s. This is formalized in Proposition~\ref{prop:flaw} (together with the argument for integrability). The flaw in the FHV-rule is that although the expected value of $x$ decreases, at each iteration we reach $x=2$ (the barrier) with higher and higher probability. So, despite the overall decrease in expectation, $x$ reaches its maximum possible value with high probability, which leads to non-termination.
}
\end{example}
\mingzhang{
We now provide a rigorous proof of the above arguments.
}

\begin{proposition}\label{prop:count1}
The probabilistic program in Example~\ref{ex:counterexample} (Figure~\ref{example:counter}, Page~\pageref{example:counter}) is not a.s. terminating. Specifically, it does not terminate with probability~$1$ when the initial value for the program variable $x$ is $1$ and the initial value for the program variable $y$ is sufficiently large.
\end{proposition}

\begin{proof}
The program does not terminate only if the value of $x$ in label $9$ is 2 after every execution of the inner loop. The key point is to prove that in the inner loop, the value of the program variable $x$ will be 2 with higher and higher probability when the value of $y$ increases.
   Consider the random walk in the inner loop. We abstract the values of $x$ as three states `$\le 0$', `$1$' and `$2$'.
   From the structure of the program, we have that if we start with the state `$1$', then after the inner loop, the successor state may transit to either `$\le 0$', `$1$' or `$2$'.
   If the successor state is either `$\le 0$' or `$1$', then the outer loop will terminate immediately. However, there is a positive probability that the successor state is `$2$' and the outer loop does not terminate in this loop iteration (as the value of $x$ will be set back to $1$).
   This probability depends on the steps of the random walk in the inner loop (determined by the value of $y$), and we show that it is higher and higher when the value of $y$ increases.
   Thus, after more and more executions of the outer loop, the value of $y$ continues to increase exponentially, and with higher and higher probability the program would be not terminating in the current execution of the loop body.

The detailed demonstration is as follows: W.l.o.g., we assume that $x=1$ at every beginning of the inner loop. The values of $x$ at label $9$ are the results of the execution of inner loop with the same initial value, hence they are independent mutually.
We now temporarily fix the value $ \hat{y}$ for $ y$ at the beginning of the outer loop body and consider the probability that the value of $x$ in label $9$ is not 2. We use the random variable $\bar{X}_{ \hat{y}}$ to describe the value of $x$ at label $9$ and analyze the situation $\bar{X}_{\hat{y}}\neq 2$ after the $\hat{y}$ iterations of the inner loop.
Suppose that the $\hat{y}$ sampled values for $r$ during the execution of the inner loop consist of $m$ instances of $-1$ and $(\hat{y} - m)$ instances of $1$. Since $\bar{X}_{\hat{y}}\neq 2$, we have  $m\geq\frac{ \hat{y}}{2}$. Then,
there are ${\hat{y}\choose{m}}-{\hat{y}\choose{m+1}}$ different possible paths that avoid being absorbed by the barrier. The reason is that the only way to avoid absorption is to always have more $-1$'s than $1$'s in any prefix of the path. Hence, the number of possible paths is the Catalan number.
 so we have
$\probm(\bar{X}_{\hat{y}} = 2 )=1-\frac{1}{2^{ \hat{y}}}\sum_{\frac{{ \hat{y}}}{2}\leq m\leq \hat{y}}({ \hat{y}\choose{m}}-{ \hat{y}\choose{m+1}}) =1-\frac{1}{2^{\hat{y}}}{{ \hat{y}}\choose{\lceil\frac{ \hat{y}}{2}\rceil}}$.
Since $\sqrt{2\pi}n^{n+\frac{1}{2}}e^{-n}\leq n! \leq e n^{n+\frac{1}{2}}e^{-n}$ for $n\geq 1$ (applying Stirling's approximation), we have $1-\frac{1}{2^{\hat{y}}}{\hat{y}\choose{\lceil\frac{\hat{y}}{2}\rceil}}=1-\frac{\hat{y}!}{2^{\hat{y}}(\frac{\hat{y}}{2}!)^2}\geq
1-\frac{e \hat{y}^{\hat{y}+\frac{1}{2}}e^{-\hat{y}}}{2^{\hat{y}}(\sqrt{2\pi}\frac{\hat{y}}{2}^{\frac{\hat{y}}{2}+\frac{1}{2}}e^{-\frac{\hat{y}}{2}})^2}
=1-\frac{e}{\pi{\sqrt{\hat{y}}}}$ for every even $\hat{y}$.
Note that $\probm(T=\infty)=\prod_{i\in\Nset_0}\probm(\bar{X}_{\hat{y}_i}=2)$, where $\hat{y}_i$ is the value of $y$ at the $i$-th arrival to the label $9$ and recall that $\hat{y}_0$ is sufficiently large.
Furthermore, from the program we have $\hat{y}_i=4^i\cdot \hat{y}_0$.
Let $d:=\frac{e}{\pi\sqrt{\hat{y}_0}}$, and we obtain that $\probm(T=\infty)=\prod_{i\in\Nset_0}\probm(\bar{X}_{\hat{y}_i}=2)=\prod_{i\in\Nset_0}(1-\frac{1}{2^{ \hat{y}_i}}{\hat{y}_i\choose{\lceil\frac{\hat{y}_i}{2}\rceil}})
\geq \prod_{i\in\Nset_0}(1-\frac{d}{\sqrt{4^i}})$.
  A well-known convergence criterion for infinite products is that $\prod_{i\in\Nset_0}(1-q_i)$ converges to a non-zero number if and only if $\sum_{i\in\Nset_0}q_i$ converges for $0\leq q_i<1$. Since $\sum_{i\in\Nset_0}\frac{d}{2^{i}}$ converges to $2d$, we have the infinite product $\prod_{i\in\Nset_0}(1-\frac{d}{2^i})$ converges to a non-zero number. Thus, $\probm(T=\infty)>0$.
\end{proof}

We now show that, using the FHV-rule proposed in \cite{HolgerPOPL}, one can deduce that the probabilistic program in Example~\ref{example:counter} is a.s. terminating.
\begin{proposition}\label{prop:flaw}
The FHV-rule in~\cite{HolgerPOPL} derives that the probabilistic program in Example~\ref{example:counter} is a.s. terminating.
\end{proposition}
\begin{proof}
To see that the FHV-rule derives a.s. termination on this example,
we show that the expression $x$ is a compositional RSM that satisfies the integrability and difference-boundedness conditions.
First, we can show that the program variable $x$ is integrable and difference-bounded at every label.
For example, for assignment statements at labels 2, 5, 7, 8 and 9 in Figure~\ref{example:counter}, the expression $x$ is integrable and difference-bounded after these statements simply because
either they do not involve $x$ at the left hand side or the assignment changes the value of $x$ by 1 unit.
Similarly, within the nested loop,
the loop body (from label 4 to label 7) causes bounded change to the value of $x$,
so the expression $x$ is integrable after the inner loop (using the while-rule in
\cite[Table 2]{HolgerPOPL}).
Second, it is easy to see that the expression $x$ is a compositional RSM as from
\cite[Definition 7.1]{HolgerPOPL}
we have the following:
\begin{compactitem}
\item The value of $x$ does not increase after the assignment statements $z:=y$ and $y:=4\times y$;
\item In the loop body of the nested loop, the expected value of $x$ does not increase, given that
it does not increase in any of the conditional branches;
\item By definition
of $Dec_{\le} (-,-)$, the expected value of $x$ does not increase after the inner loop;
\item The value of $x$ is decreased by $1$ after the last assignment statement $x:=x-1$.
\end{compactitem}
Thus, by applying~\cite{HolgerPOPL}'s main theorem for compositionality
(\cite[Theorem 7.7]{HolgerPOPL}),
we can conclude that the program should be a.s. terminating.
\end{proof}

From Proposition~\ref{prop:count1} and Proposition~\ref{prop:flaw}, we establish the main theorem of this section, i.e.~that the FHV-rule is not sound.

\begin{theorem}\label{thm:unsoundness}
The FHV-rule, i.e.~the probabilistic extension of the V-rule as proposed in \cite{HolgerPOPL}, is not sound for a.s. termination of probabilistic programs, even if we require the compositional RSM $R$ to be difference-bounded and integrable.
\end{theorem}

Note that integrability is a very natural requirement in probability theory. Hence, Theorem~\ref{thm:unsoundness} states that a natural probabilistic extension of the variant rule is not sufficient for proving a.s. termination of probabilistic programs.

\paragraph{Crucial Issue 2: Non-negativity of RSMs.} \label{cruc2}
	The reason why the approach of \cite{HolgerPOPL} is not sound lies in the fact that their approach neglects the non-negativity of ranking supermartingales (RSMs), \hongfei{as their compositional RSMs are not required to be non-negative}.
	In the classical V-rule for non-probabilistic programs, non-negativity is not required, given that negative values in a non-probabilistic setting simply mean that $R$ is negative.
	However, in the presence of probability, the expected value of $R$ is taken into account, and not $R$ itself. Thus, it is possible that the expected value of $R$ decreases and
	becomes arbitrarily negative, tending to $-\infty$, while simultaneously the value of $R$ increases with higher and higher probability.
	In our counterexample (Example~\ref{ex:counterexample}), the expected value of $x$ decreases after each outer loop iteration, however
	the probability that the value of $x$ remains the same increases with the value of $y$. More specifically, the decrease in the expected value results from the fact that after the inner loop,
	the value of $x$ may get arbitrarily negative towards $-\infty$.
	For a detailed explanation of the unsoundness of the FHV-rule, see Appendix~\ref{app:flaw}.

\section{Our Modular Approach}\label{sec:concentration}

In the previous section, we showed that the FHV-rule is not sound for proving a.s. termination of probabilistic programs.
In this section, we show how the FHV-rule can be \hongfei{minimally} strengthened to a sound approach.
The general idea of our approach is to define a new notion called a \emph{Descent Supermartingale Map (DSM)} that requires the expected value of the expression $R$ in the variant rule to always decrease by at least a positive amount $\epsilon$. We call this the \emph{strict decrease} condition. This condition is in contrast with the FHV-rule that allows the value of $R$ at certain statements to remain the same (in expectation).
We show that after this strengthening, the resulting rule is sound for \amir{modular} verification of  a.s. termination over probabilistic programs. \amir{We handle Crucial Issue~1 in the same manner as in the FHV-rule, i.e.~by enforcing difference-boundedness. As for Crucial Issue~2, we show that unlike RSMs, DSMs do not require non-negativity. Hence, this issue does not apply to our approach.}
Our main mathematical tools are the \emph{concentration inequalities} (e.g.~\cite{ColinMcDiarmid1998concentration}) that give tight upper bounds on the probability that a stochastic process deviates from its mean value.

\amir{In this section, we present our approach in terms of martingales and prove its soundness. In the next sections, we provide both a proof system based on inference rules and a synthesis algorithm based on templates for obtaining DSM-maps. This shows that our approach (i)~leads to a completely automated method for proving a.s.~termination (similar to other martingale-based approaches such as~\cite{SriramCAV,ChatterjeeFNH16,ChatterjeeFG16}) and (ii)~can also be applied in a semi-automatic setting using interactive theorem provers (similar to other rule-based approaches such as~\cite{KaminskiKMO16,DBLP:conf/lics/OlmedoKKM16,DBLP:journals/pacmpl/McIverMKK18}). Hence, our approach combines the best aspects of both martingale-based and rule-based approaches, with the added benefit of being modular}.

To clarify that our approach is indeed a strengthening of the FHV-rule in~\cite{HolgerPOPL},
we first write the rule-based approach of~\cite{HolgerPOPL} in an equivalent martingale-based format.
Below, we fix a probabilistic program $P'$ and a loop guard $G$ and let $P:=\textbf{while}(G,P')$.
For the purpose of \amir{modular} verification, we assume that $P'$ is a.s. terminating.
We recall that $T$ is the termination-time random variable (see Definition~\ref{def:astermination}) and $\sampdpd$
is the joint discrete probability distribution for sampling variables.
We also use the standard notion of invariants, which are over-approximations of the set of reachable configurations at every label.

\paragraph{Invariants.} An \emph{invariant} is a function $I: L \rightarrow 2^{\val{V_p}}$, such that for each label $\loc \in L$,
the set $I(\loc)$ at least contains all valuations $\nu$ of program variables for which the configuration $(\loc,\nu)$ can be visited in some run of the program.
An invariant $I$ is \emph{linear} if every $I(\loc)$ is a finite union of polyhedra.

We can now describe the FHV-rule approach in~\cite{HolgerPOPL} using \emph{V-rule supermartingale maps}.
A \emph{V-rule supermartingale map} w.r.t. %
an invariant $I$ is a function $R: \val{\pvars}\rightarrow\Rset$  satisfying
the following conditions:

\begin{compactitem}
\item {\em Non-increasing property.} The value of $R$ does not increase in expectation after the execution of any of the statements in the outer-loop body.
For example, the non-increasing condition for an assignment statement $\loc\in\alocs$ with $(\loc,u,\loc')\in\transitions$ (recall that $u$ is the update function) is equivalent to
$\sum_{\mu\in \val{\rvars}} \sampdpd(\mu)\cdot R(u(\nu,\mu))\leq  R(\nu) $ for all $\nu\in I(\loc)$.
This condition can be similarly derived for other types of labels.

\item {\em Decrease property.} There exists a statement that will definitely be executed in every loop iteration and will cause $R$ to decrease (in expectation). For example, the condition for strict decrease at an assignment statement $\loc\in\alocs$ with $(\loc,u,\loc')\in\transitions$ says that for all $\nu\in I(\loc)$ we have
  $\sum_{\mu\in \val{\rvars}} \sampdpd(\mu)\cdot R(u(\nu,\mu))\leq R(\nu)-\epsilon $, where $\epsilon$ is a fixed positive constant.
 \item {\em Well-foundedness.} The values of $R$ should be bounded from below when restricted to the loop guard.
Formally, this condition requires that for a fixed constant $c$ and all $\nu\in I(\lin)$ such that $\nu\models G$, we have $R(\nu)\ge c$.
\item {\em Conditional difference-boundedness.} The conditional expected change in the value of $R$ after the execution of each statement is bounded. For example, at an assignment statement $\loc\in\alocs$ with $(\loc,u,\loc')\in\transitions$, this condition says that there exists a fixed positive bound $d$, such that $ \sum_{\mu\in \val{\rvars}} \sampdpd(\mu)\cdot |R(u(\nu,\mu))-R(\nu)|\le d$ for all $\nu\in I(\loc)$. The purpose of this condition is to ensure the integrability of $R$ (see~\cite[Lemma 7.4]{HolgerPOPL}).
\end{compactitem}

\paragraph{Strengthening.} We strengthen the FHV-rule of~\cite{HolgerPOPL} in two ways.
First, as the major strengthening, we require that the expression $R$ should strictly decrease in expectation at every statement,
as opposed to~\cite{HolgerPOPL} where the value of $R$ is only required to decrease at some statement.
Second, we slightly extend the conditional difference-boundedness condition and require that
the difference caused in the value of $R$ after the execution of each statement should always be bounded, i.e.~we require difference-boundedness not only in expectation, but in every run of the program.

The core notion in our strengthened approach is that of \emph{Descent Supermartingale maps (DSM-maps)}. A DSM-map is a function
representing a decreasing amount (in expectation) at each step of the execution of the program.

\begin{definition}[Descent Supermartingale Maps]\label{def:dbprsm}
A \emph{descent supermartingale map} (DSM-map) w.r.t. real numbers $\epsilon>0$, $c\in\Rset$, a non-empty interval $[a,b]\subseteq\mathbb{R}$ and an invariant $I$ is a function $\eta: \locs\times\val{\pvars}\rightarrow\Rset$  satisfying
the following conditions:
\begin{compactitem}
 \item [(D1)] For each $\loc\in\alocs$ with $(\loc,u,\loc')\in\transitions$ , it holds that
\begin{compactitem}
\item  $a\le \eta(\loc',u(\nu,\mu))-\eta(\loc,\nu)\le b$ for all $\nu\in I(\loc)$ and $\mu\in\val{\rvars}$;
\item $\sum_{\mu\in \val{\rvars}} \sampdpd(\mu)\cdot \eta(\loc',u(\nu,\mu))\leq \eta (\loc,\nu)-\epsilon $ for all $\nu\in I(\loc)$;
\end{compactitem}
\item [(D2)] For each $\loc\in\blocs$ and $(\loc,\phi,\loc')\in\transitions$, it holds that $a\le \eta(\loc',\nu)-\eta(\loc,\nu)\le \min\{-\epsilon,b\}$ for all $\nu\in I(\loc)$ such that $\nu\models \phi$;
\item [(D3)] For each $\loc\in\dlocs$ and $(\loc,\star,\loc')\in\transitions$, it holds that $a\le \eta(\loc',\nu)-\eta(\loc,\nu)\le \min\{-\epsilon,b\}$ for all $\nu\in I(\loc)$;
\item [(D4)] For each $\loc\in\plocs$ with $(\loc,p,\loc'),(\loc,1-p,\loc'')\in\transitions$, it holds that
  \begin{compactitem}
    \item $a\le \eta(\loc',\nu)-\eta(\loc,\nu)\le b$ for all $\nu\in I(\loc)$,
    \item $a\le \eta(\loc'',\nu)-\eta(\loc,\nu)\le b$ for all $\nu\in I(\loc)$,
    \item $p\cdot\eta(\loc',\nu)+(1-p)\cdot \eta(\loc'',\nu)\le \eta (\loc,\nu)-\epsilon$ for all $\nu\in I(\loc)$;
  \end{compactitem}
\item [(D5)]  For all $\nu\in I(\lin)$ such that $\nu\models G$
(recall that $G$ is the loop guard),
it holds that $\eta(\lin,\nu)\ge c$.
\end{compactitem}
\end{definition}

\noindent Informally, $R$ is a DSM-map if:
\begin{compactitem}
\item[(D1)--(D4)] Its value decreases in expectation by at least $\epsilon$ after the execution of each statement (the strict decrease condition), and
its change of value before and after each statement falls in $[a,b]$ (the strengthened difference-boundedness condition);
\item[(D5)] Its value is bounded from below by $c$ at every entry into the loop body (the well-foundedness condition).
\end{compactitem}

\amir{
\begin{remark} \label{rem:d5}
	We remark two points about DSM-maps:
	\begin{compactitem}
		\item DSM-maps require well-foundedness only w.r.t. the outermost loop guard. See (D5) above.
		\item The function $\eta$ is dependent not only on the valuation, but also on the label (program counter). Hence, it can correspond to different expressions at each label. Informally, we do not have a single fixed expression $R$, but instead have a label-dependent expression $\eta(\loc)$. This gives our approach more flexibility.
	\end{compactitem}
\end{remark}
}

By the decreasing nature of DSM-maps, it is intuitively true that the existence of a DSM-map implies a.s. termination.
However, this point is non-trivial as counterexamples will arise if we drop the difference-boundedness condition and only require the strict decrease condition
(see e.g.~\cite[Example~3]{DBLP:conf/aplas/HuangFC18}).
In the following, we use the difference-boundedness condition to derive a concentration property on the termination time (see~\cite{ChatterjeeFNH16}).
Under this concentration property, we prove that DSM-maps are sound for proving a.s. termination.

We first present a well-known concentration inequality called \emph{Hoeffding's Inequality}.

\begin{theorem*}[Hoeffding's Inequality~\cite{Hoeffding1963inequality}]\label{thm:hoeffding}
Let $\{X_n\}_{n\in\mathbb{N}_0}$ be a supermartingale w.r.t. some filtration $\{\mathcal{F}_n\}_{n\in\mathbb{N}}$ and $\{[a_n,b_n]\}_{n\in\mathbb{N}}$ be a sequence of intervals with positive length in $\mathbb{R}$.
If $X_0$ is a constant random variable and $X_{n+1}-X_n\in [a_n,b_n]$ a.s. for all $n\in\mathbb{N}_0$, then
\[
\mathbb{P}(X_n-X_0\ge\lambda)\le \exp(-\frac{2\lambda^2}{\sum_{k=1}^n(b_k-a_k)^2})
\]
for all $n\in\mathbb{N}_0$ and $\lambda> 0$.
\end{theorem*}

Hoeffding's Inequality states that for any difference-bounded supermartingale, it is unlikely that
its value $X_n$ at the $n$-th step exceeds its initial value $X_0$ by much (measured by $\lambda$).

Using Hoeffding's Inequality, we prove the following lemma.

\begin{lemma}~\label{lemm:concen}
  Let $\{X_n\}_{n\in\mathbb{N}_0}$ be a supermartingale w.r.t. some filtration $\{\mathcal{F}_n\}_{n\in\mathbb{N}}$ and $[a,b]$ be an interval with positive length in $\Rset$.
 If $X_0$ is a constant random variable, it holds that $\condexpv{X_{n+1}}{\mathcal{F}_n}\le X_n-\epsilon$   for some $\epsilon>0$ and $X_{n+1}-X_n\in[a,b]$ a.s. for all $n\in\mathbb{N}_0$, then
for any $\lambda\in\Rset$,
\[
\mathbb{P}(X_n-X_0\ge\lambda)\le {\exp({-\frac{2(\lambda+n\cdot\epsilon)^2}{n(b-a)^2}})}
\]
for all sufficiently large $n$.
\end{lemma}
\begin{proof}
   Let $Y_n=X_n+n\cdot\epsilon$, then $a+\epsilon\leq Y_{n+1}-Y_{n}=X_{n+1}-X_{n}+\epsilon\leq b+\epsilon$.
   Given that
  \begin{eqnarray*}
      \condexpv{Y_{n+1}}{\mathcal{F}_n}&=& \condexpv{X_{n+1}}{\mathcal{F}_n}+(n+1)\cdot\epsilon\\
      &\leq& X_n+n\cdot\epsilon\\
      &=&Y_n,
    \end{eqnarray*}
    we conclude that $\{Y_n\}_{n\in\Nset_0}$ is a supermartingale.
   Now we apply Hoeffding's Inequality for all $n$ such that $\lambda+n\cdot\epsilon>0$, and we get
    \begin{eqnarray*}
   \probm(X_n-X_0\geq \lambda)&=&\probm(Y_n-Y_0\geq \lambda+n\cdot\epsilon)\\
   &\leq& { \exp({-\frac{2(\lambda+n\cdot\epsilon)^2}{n(b-a)^2}})}\\
   \end{eqnarray*}
\end{proof}

Thus, we have the following corollary by calculation.
\begin{corollary}~\label{cor:lim}
  Let $\{X_n\}_{n\in\Nset_0}$ be a supermartingale satisfying the conditions of Lemma~\ref{lemm:concen}. Then, $\lim_{n\rightarrow +\infty}\sum_{k=n}^{+\infty}\probm(X_k-X_0\geq \lambda)=0$.
\end{corollary}

We are now ready to prove the soundness of DSM-maps.

\begin{theorem}[Soundness of DSM-maps]~\label{thm:rsm}
Let $P=\textbf{while}(G,P').$ If (i)~$P'$ terminates a.s.~for any initial valuation and scheduler; and (ii)~there exists a DSM-map $\eta$ for $P$, then for any initial valuation $\nu^{*}\in\val{\pvars}$ and for all schedulers $\sigma$,
we have $\probm_{\nu^{*}}^\sigma(T< \infty)=1$.
\end{theorem}

\begin{proof}[Proof Sketch]
Let $\epsilon,c,a,b$ be as defined in Definition~\ref{def:dbprsm}.
For a given program $P$ with its DSM-map $\eta$, we define the stochastic process $\{X_n=\eta(\loc_n,\nu_n)\}_{n\in \Nset_0}$ where $(\loc_n,\nu_n)$ is the pair of random variables that represents the configuration at the $n$-th step of a run. We also define
the stochastic process $\{B_n\}_{n\in \Nset_0}$ in which
each $B_n$ represents the number of steps in the execution of $P$ until the $n$-th arrival at the initial label $\lin$. Then, $X_{B_n}$ is the random variable representing the value of $\eta$ at the $n$-th arrival at $\lin$.
Recall that, by condition (D5) in the definition of DSM-maps, the program stops if $X_{B_n}<c$. We now prove the crucial property that
$\probm(T'< \infty)\geq1-\lim_{n\rightarrow\infty}\probm(X_{B_n}\ge c)=1$, where $T'$ is the random variable that measures the number of outer loop iterations in a run.
We want to estimate the probability of $\probm(X_{B_n}\ge c)$ which is bounded by $\sum_{k=n}^{+\infty}\probm(X_k\geq c)$.
Note that $X_n$ satisfies the conditions of Lemma~\ref{lemm:concen}. We use Corollary~\ref{cor:lim} to bound the probability.
Since $\probm(T<\infty)=1$ iff $\probm(T'<\infty)=1$ (as $P'$ is a.s. terminating), we obtain that $\probm(T< \infty)=1$.
For a more detailed proof, see Appendix~\ref{app:thm:rsm}.

\end{proof}

\begin{remark}[{\hongfei{Modularity}}]
	The theorem above directly leads to a \amir{modular} approach for proving a.s.~termination. To prove that $P$ terminates a.s., it suffices to first prove that $P'$ terminates a.s., and then show the existence of a DSM-map w.r.t. $P$ as \hongfei{the} side condition.
\amir{
	We stress that in the theorem above, it is necessary to assume that $P'$ terminates a.s., because DSM-maps consider well-foundedness and termination only w.r.t. the outermost loop guard (see Remark~\ref{rem:d5}). Hence, the existence of a DSM-map does not prove a.s.~termination in and of itself, but only serves as a side condition in our modular approach.}
\end{remark}

We illustrate an example application of Theorem~\ref{thm:rsm}.
\begin{example}\label{example:rsp}
Consider the probabilistic while loop in Figure~\ref{fig:exprog}. %
\begin{figure}
\begin{center}
\begin{tabular}[t]{l}
\lstset{language=prog}
\begin{lstlisting}[mathescape, basicstyle=\small]
$\ \,1:$ while $x\geq 1$ do
$\ \,2:$     $y:=r;$
$\ \,3:$     while $y\geq1$ do
$\ \,4:$        if $\star$ then
$\ \,5:$           if prob (6/13) then
$\ \,6:$               $x:=x+1$
$\,\ \,\,$           else
$\,\ 7:$               $x:=x-1$
$\,\ \,\,$           fi
$\,\ \,\,$        else
$\,\ 8:$           if prob (4/13) then
$\,\ 9:$               $x:=x+2$
$\,\ \,\,$           else
$10:$               $x:=x-1$
$\,\ \,\,$           fi
        fi;
$11:$        $y:=y-1$
      od
    od
$12:$
\end{lstlisting}
\end{tabular}
\end{center}
\vspace{-.6em}
\caption{An Example Probabilistic Program. In this program, $\probm(r = k) = 1/9$ for $k = 1, 2, \ldots, 9.$}
\vspace{-1em}
\label{fig:exprog}
\end{figure}
\noindent where the probability distribution for the sampling variable $r$ is
given by $\probm(r=k)=1/9$ for $k=1,2,\ldots,9$.

The while loop models a variant of gambler's ruin based on the mini-roulette game with $13$ slots~\cite{DBLP:conf/ijcai/ChatterjeeFGO18}.
Initially, the gambler has $x$ units of money and he continues betting until he has no money.
At the start of each outer loop iteration, the number of gambling rounds is chosen uniformly at random from $1,2,\dots,9$ (i.e. the program variable $y$ is the number of gambling rounds in this iteration).
Then, at each round, the gambler takes one unit of money, and either chooses an \emph{even-money bet} that bets the ball to stop at even numbers between $1$ and $13$, which has a probability of $\frac{6}{13}$ to win one unit of money (see the non-deterministic branch from label $5$ to label $7$), or a \emph{$2$-to-$1$ bet} that bets the ball to stop at $4$ selected slots and wins two units of money with probability $\frac{4}{13}$ (see the branch from label $8$ to label $10$).
During each outer loop iteration, it is possible that the gambler runs out of money temporarily, but the gambler is allowed to continue gambling in the current loop iteration,
and the program terminates only if he depletes his money when the program is back to the start of the outer loop.
An invariant $I$ for the program is as follows:
$$\quad   I(\loc):=\begin{cases}
         \textbf{true}  & \mbox{ if } \loc =1\\
     x\geq 1     & \mbox{ if } \loc =2\\
     x\geq -8\wedge 0\leq y\leq 9   & \mbox{ if } \loc =3\\
     x\geq -7\wedge 1\leq y\leq 9   & \mbox{ if } 4\leq\loc\leq11\\
   \end{cases}\quad.$$
For this program, we can define a DSM-map $\eta$ as follows:
\[
    \eta(\loc, (x,y)):=\begin{cases}
    x  & \mbox{ if } \loc =1\\
     x-4/299    & \mbox{ if } \loc =2,12\\
     x-3/299\cdot y+ 7/299  & \mbox{ if } \loc=3\\
       x-3/299\cdot y+ 3/299    & \mbox{ if } \loc=4\\
        x-3/299\cdot y- 1/299    & \mbox{ if } \loc=5,8\\
     x-3/299\cdot y+ 317/299  & \mbox{ if } \loc =6\\
     x-3/299\cdot y- 281/299  & \mbox{ if } \loc =7,10\\
     x-3/299\cdot y+ 616/299  & \mbox{ if } \loc =9\\
     x-3/299\cdot y+ 14/299 & \mbox{ if } \loc =11\\
    \end{cases}.
\]
One can verify that $\eta$ is a DSM-map by choosing $\epsilon=4/299, a=-280/299, b=617/299$ and $c=1$. The minimal and maximal one-step differences of $\eta$
are met in the transitions from labels $9$ and $10$ to label $11$. Thus, the differences are in the interval  $[-1+19/299,2+19/299]=[a,b]$, and the expected value of $\eta$ decreases by at least $4/299=\epsilon$ in each step. Also, if the outer loop is not stopped, then $x\geq1=c$ at the initial label.
The other conditions can be similarly checked. Thus, $\eta$ is a DSM-map for $P$. Note that the internal while loop is a classic gambler's ruin and it is well-known that it terminates a.s. Therefore, by applying Theorem~\ref{thm:rsm}, we conclude that the program terminates a.s. under any initial valuation.
\end{example}

{We now compare the notion of DSM-maps with RSMs/RSM-maps~\cite{SriramCAV,ChatterjeeFNH16,ChatterjeeFG16} that have been successfully applied to prove finite expected termination time of probabilistic programs.}

\begin{remark}[Comparison with RSMs]
Our notion of DSM-maps is slightly (but crucially) different from RSMs~\cite{SriramCAV,ChatterjeeFNH16,ChatterjeeFG16}.
The difference is that a DSM-map does not require a global lower bound on its values,
but instead requires the difference-boundedness condition, while an RSM requires its values to be non-negative, but has no difference-boundedness condition.
\hongfei{As a result, the soundness of DSM-maps follows from concentration inequalities, while the soundness of RSMs follows from a limiting behavior of non-negative stochastic processes, a completely different aspect.}

\end{remark}

\begin{remark}[Comparison with~\cite{HolgerPOPL}] We remark the reason why the approach in \cite{HolgerPOPL} is not sound while ours is. This has to do with Crucial Issue~2 (Page~\pageref{cruc2}). The approach in~\cite{HolgerPOPL} neglects the fact that RSMs have to be non-negative and is therefore not sound. In contrast, our approach uses DSM-maps which are not restricted to be non-negative and are sound for proving a.s. termination of probabilistic programs.
As we have described previously, our approach of DSM-maps mainly strengthens the approach in \cite{HolgerPOPL} with the strict decrease condition at every statement.
\hongfei{The negativity of DSM-maps is then resolved}
through concentration inequalities, which guarantee that the probability of the value of $R$ tending unboundedly to $-\infty$ is exponentially decreasing (see Lemma~\ref{lemm:concen}).

\hongfei{
In~\cite{HolgerPOPL}, the authors \amir{describe} their approach \amir{as} \emph{compositional}. However, %
it \amir{requires} compositional RSMs as a non-trivial side condition. In this work, we call such approaches \emph{modular}.} \amir{Note that our side conditions (DSMs) have the same modularity and complexity as the side conditions required by~\cite{HolgerPOPL} (compositional RSMs).}
\end{remark}

\mingzhang{
\begin{remark}[DSM-maps as a Debugging Tool]
	Note that the DSM-maps contain a lot of useful information about the programs. For example, the verification could be ``inverted'' and used as a debugging tool, since DSM-maps can provide witnesses for proving/refuting a.s. termination. This point has also been mentioned in previous results on RepSMs (See~\cite{ChatterjeeNZ2017}).  Given that our approach is modular and hence much faster for larger programs, it can also be used as an almost-real-time debugging tool and this point applies to it even more strongly than previous RepSM approaches.
\end{remark}
}

\begin{remark}[Real-valued Variables]
Although we illustrate our approach on integer-valued variables, we show that it also works for real-valued variables.
First, we directly extend the notion of DSM-maps to real-valued variables, where we only replace
the discrete summation $\sum_{\mu\in \val{\rvars}} \sampdpd(\mu)\cdot \eta(\loc',u(\nu,\mu))$ to an integral.
Then we can prove the soundness of DSM-maps and construct the synthesis algorithm in the same way as for the integer case.
\end{remark}

\section{A Sound Proof System for Almost-Sure Termination}

In this section, we provide a proof system $\mathcal{D}$, in the style of Hoare logic, for the DSM approach. Note that DSM-maps treat the outer loop in a different manner than the inner loops. Specifically, in Definition~\ref{def:dbprsm}, the requirement (D5) is only applied to the outer loop. This distinction is handled in our proof system by introducing different kinds of Hoare triples.
Let $R$ and $R'$ be arithmetic  expressions over program variables. The Hoare triple $\{R \}P \{R'\}$ indicates that for the program $P$, there exists a DSM-map $\eta$ such that $\eta(\lin^{P},\b{ } )=R$ and $\eta(\lout^P, \b{ })=R'$. Similarly, the triple $\langle R \rangle P \langle R'\rangle$ indicates the existence of a map $\eta$ with $\eta(\lin^{P},\b{ } )=R$ and $\eta(\lout^P, \b{ })=R'$ that satisfies all conditions of a DSM-map except for (D5). Intuitively, such an $\eta$ is not a DSM-map for $P$, but it can be extended to a DSM-map for another program that contains $P$ as an inner loop. We use the term \emph{partial DSM} to describe such $\eta$.

We have following axiom schemata and rules in $\mathcal{D}$. Note that the \amir{DSM while rules are the most important novelty in our proof system. They are also the only pair of rules that are different for the two kinds of Hoare triples, i.e.~the first while rule requires (D5), while the second while rule does not.} Moreover, the values of $a, b, c$ and $\epsilon$ are not fixed and can be different in every application of the rules below.
\begin{compactenum}
	\item DSM while rules:
	$$\inference{\langle R \rangle P \langle R'\rangle, \mathbf{G \rightarrow R'\geq c},
		\\G \rightarrow  a\leq R-R'\leq-\epsilon,\\ \text{ and } \neg G \rightarrow a\leq R''- R'\leq-\epsilon
	}{\{R'\}\textbf{ while } G \textbf{ do } P\textbf{ od }\{ R''\}} \quad, \quad
	\inference{\langle R \rangle P \langle R'\rangle,
		\\G \rightarrow  a\leq R-R'\leq-\epsilon,\\ \text{ and } \neg G \rightarrow a\leq R''- R'\leq-\epsilon
	}{\langle R'\rangle \textbf{ while } G \textbf{ do } P\textbf{ od }\langle R''\rangle}$$
	\item Skip statement axiom schemata:
	$$\inference{a\leq R'-R\leq-\epsilon}{\{R\}\textbf{ skip }\{R'\}} \quad, \quad \inference{a\leq R'-R\leq-\epsilon}{\langle R \rangle\textbf{ skip }\langle R'\rangle}$$
	\item Assignment axiom schemata:
	$$\inference{a\leq \expv[R'[x \leftarrow \mathfrak{e}]]-R\leq-\epsilon}{\{R\}~x:=\mathfrak{e} ~\{R'\}} \quad, \quad
	\inference{a\leq \expv[R'[x \leftarrow \mathfrak{e}]]-R\leq-\epsilon}{\langle R\rangle~x:=\mathfrak{e} ~\langle R'\rangle}
	$$
	Here $R'[x \leftarrow \mathfrak{e}]$ is the expression obtained when one replaces all occurrences of the variable $x$ in $R'$ by the expression $\mathfrak{e}.$
	\item Sequential composition rules:
	$$\inference{\{R\}P_1\{R'\},\{R'\}P_2\{R''\}}{\{R\}P;Q\{R''\}} \quad, \quad
	\inference{\langle R\rangle P_1 \langle R'\rangle,\langle R'\rangle P_2 \langle R''\rangle}{\langle R\rangle P;Q\langle R''\rangle}$$
	\item Conditional branch rules:
	$$\inference{\{R_1\}P_1\{R'\},\{R_2\}P_2\{R'\},\\G \rightarrow a\leq R_1-R\leq-\epsilon,\\ \text{ and } \neg G \rightarrow a\leq R_2-R\leq-\epsilon}{\{R\}\textbf{ if } G \textbf{ then }P_1\textbf{ else }P_2\{R'\}} \quad, \quad
	\inference{\langle R_1\rangle P_1\langle R'\rangle,\langle R_2\rangle P_2\langle R'\rangle,\\G \rightarrow a\leq R_1-R\leq-\epsilon,\\ \text{ and } \neg G \rightarrow a\leq R_2-R\leq-\epsilon}{\langle R\rangle\textbf{ if } G \textbf{ then }P_1\textbf{ else }P_2\langle R'\rangle}
	$$
	\item Non-deterministic branch rules:
	$$\inference{\{R_1\}P_1\{R'\},\{R_2\}P_2\{R'\},\\ a\leq R_1-R\leq-\epsilon,\\ \text{ and } a\leq R_2-R\leq-\epsilon} {\{R\}\textbf{ if } \star \textbf{ then }P_1\textbf{ else }P_2\{R'\}} \quad, \quad
	\inference{\langle R_1\rangle P_1\langle R'\rangle,\langle R_2\rangle P_2\langle R'\rangle,\\ a\leq R_1-R\leq-\epsilon,\\ \text{ and } a\leq R_2-R\leq-\epsilon} {\langle R\rangle\textbf{ if } \star \textbf{ then }P_1\textbf{ else }P_2\langle R'\rangle}
	$$
	\item Probabilistic  branch rules:
	$$\inference{\{R_1\}P_1\{R'\},\{R_2\}P_2\{R'\},\\ a\leq R_1-R\leq b,\\ a\leq R_2-R\leq b,
		\\ \text{ and } p\cdot R_1+(1-p)\cdot R_2 \leq R-\epsilon
	} {\{R\}\textbf{ if } \textbf{prob}($p$) \textbf{ then }P_1\textbf{ else }P_2\{R'\}} \quad, \quad
	\inference{\langle R_1\rangle P_1\langle R'\rangle,\langle R_2\rangle P_2\langle R'\rangle,\\ a\leq R_1-R\leq b,\\ a\leq R_2-R\leq b,
	\\ \text{ and } p\cdot R_1+(1-p)\cdot R_2 \leq R-\epsilon
} {\langle R\rangle \textbf{ if } \textbf{prob}($p$) \textbf{ then }P_1\textbf{ else }P_2\langle R'\rangle}$$

\end{compactenum}

The rules above can be used for establishing the existence of a DSM-map, which serves as a side condition in our \amir{modular} approach for proving a.s.~termination. Let $\tm(P)$ denote that the program $P$ terminates. The proof system $\mathcal{D}$ contains the following schemata and rules to \amir{modularly} prove a.s.~termination of programs:

\begin{compactenum}
	\setcounter{enumi}{7}
	\item \amir{Modular} DSM termination rule:
	$$
	\inference{\tm(P) \text{ and }\\ \{R_1\} \textbf{ while } G \textbf{ do } P \textbf{ od } \{R_2\}}{ \tm(\textbf{ while } G \textbf{ do } P \textbf{ od })}
	$$
	\item Assignment and skip termination axiom schemata:
	$$
	\inference{}{\tm( x := \mathfrak{e} )}
	\quad, \quad
	\inference{}{\tm(\textbf{ skip })}
	$$
	\item Sequential composition termination rule:
	$$
	\inference{\tm(P_1) \text{ and } \tm(P_2)}{\tm(P_1; P_2)}
	$$
	\item Branching composition termination rules:
	$$
	\inference{\tm(P_1) \text{ and } \tm(P_2)}{\tm(\textbf{ if } G \textbf{ then }P_1 \textbf{ else } P_2 \textbf{ fi })}
	\quad, \quad
	\inference{\tm(P_1) \text{ and } \tm(P_2)}{\tm(\textbf{ if } \star \textbf{ then }P_1 \textbf{ else } P_2 \textbf{ fi })} $$
	$$
	\inference{\tm(P_1) \text{ and } \tm(P_2)}{\tm(\textbf{ if prob} (p) \textbf{ then }P_1 \textbf{ else } P_2 \textbf{ fi })}
	$$
\end{compactenum}

\begin{theorem}\label{thm:dsm-V-rule}
 The proof system $\mathcal{D}$ is sound for a.s. termination of probabilistic programs.
\end{theorem}
\begin{proof}
  Consider the first DSM while rule. Let $Q=\textbf{while } G \textbf{ do } P\textbf{ od }$.
  Suppose that $\langle R \rangle P \langle R'\rangle$ and the partial DSM of $P$ is $\eta$, then we construct a DSM $\eta'$ by defining  $\eta'(\lin^{Q},\b{ } ):=R'$, $\eta'(\lout^{Q},\b{ } ):=R''$ and $\eta'(\loc,\b{ } ):=\eta(\loc,\b{ } )$ for all labels $\loc$ in the loop body. It is easy to check that $\eta'$ is a valid DSM, and we have $\{R'\} Q \{ R''\}$. The soundness of the other DSM while rule can be proven in a similar manner. Rules~(2)--(7) correspond to requirements (D1)--(D4) in the definition of a DSM-map (Definition~\ref{def:dbprsm}). Note that it does not matter if different $\epsilon$ values were used to obtain the preconditions of these rules, given that one can use the smallest $\epsilon$ as the parameter for the DSM. The same point applies to $a, b, c$.
   Rule~(8) is the same as Theorem~\ref{thm:rsm}. Rule~(9) is sound because a single assignment or skip statement a.s.~terminates. Soundness of rules~(10)--(11) is proven in Lemma~\ref{lem:seqcond}.
\end{proof}

We now \hongfei{argue} \hongfei{from the perspective of proof rules} that the approach of DSM-maps is a \amir{modular} approach for proving a.s.~termination.
	Note that in rule~(8) above, we use the assumption that $P$ terminates a.s., together with the side condition $\{R_1\} \textbf{ while } G \textbf{ do } P \textbf{ od } \{R_2\}$ (in the sense of Definition~\ref{def:compo}) to prove that $\textbf{ while } G \textbf{ do } P \textbf{ od }$ terminates a.s.~as well. Also, it is worth mentioning that our approach does not synthesize a global RSM for the entire program. Instead, it finds distinct DSMs for each of the while loops. 
See Appendix~\ref{app:psDSM} for a detailed example, in which two different DSM-maps are used for the internal and external while loops.

\section{the Template-Based Algorithm for Synthesizing DSM-maps}\label{sec:alg}

In this section, we provide an efficient template-based algorithm for synthesizing linear DSM-maps.
\mingzhang{
In theory, DSM-maps can take any general form. The only requirement is that they should satisfy (D1)-(D5) as in Definition~\ref{def:dbprsm}. In this section, to synthesize a DSM-map from a given program, we first assume that the function has a special form, i.e. it is linear, and then establish constraints over its coefficients. Finally, the constraints can be solved through linear programming, leading to a sound method for generation of DSM-maps.
}

Recall that the existence of DSM-maps is used as a side condition in our \amir{modular} approach for proving a.s.~termination. Concretely, the algorithm provided in this section can replace rules (1)--(7) in the proof system $\mathcal{D}$. Hence, combining it with rules (8)--(11) leads to an efficient and completely automated \amir{modular} method for proving a.s.~termination.

Since DSM-maps are similar to RSM-maps~\cite{SriramCAV,ChatterjeeFNH16,ChatterjeeFG16},
we can directly extend previous algorithms for synthesizing linear/polynomial RSM-maps~\cite{SriramCAV,ChatterjeeFNH16,ChatterjeeFG16} to  linear DSM-maps.
The key mathematical tool used in our algorithm is the well-known Farkas' Lemma.

\begin{theorem*}[Farkas' Lemma~\cite{FarkasLemma,SchrijverPolyhedra}]\label{thm:farkas}
Let $\mathbf{A}\in\mathbb{R}^{m\times n}$, $\mathbf{b}\in\mathbb{R}^m$, $\mathbf{c}\in\mathbb{R}^{n}$ and $d\in\mathbb{R}$.
Assume that $\{\mathbf{x}\in\mathbb{R}^n\mid \mathbf{A}\mathbf{x}\le \mathbf{b}\}\ne\emptyset$.
Then
\[
\{\mathbf{x}\in\mathbb{R}^n\mid \mathbf{A}\mathbf{x}\le \mathbf{b}\}\subseteq \{\mathbf{x}\in\mathbb{R}^n\mid \mathbf{c}^{\mathrm{T}}\mathbf{x}\le d\}
\]
iff there exists $\mathbf{y}\in\mathbb{R}^m$ such that $\mathbf{y}\ge \mathbf{0}$, $\mathbf{A}^\mathrm{T}\mathbf{y}=\mathbf{c}$ and $\mathbf{b}^{\mathrm{T}}\mathbf{y}\le d$.
\end{theorem*}

\paragraph{The Farkas' Linear Assertions $\Phi$.}
Farkas' Lemma transforms the inclusion testing of systems of linear inequalities into an emptiness problem.
Given a polyhedron $H=\{\mathbf{x}\in\mathbb{R}^n\mid \mathbf{A}\mathbf{x}\le \mathbf{b}\}$ as in the statement of Farkas' Lemma (Theorem~\ref{thm:farkas}),
we define the predicate $\Phi[H,\mathbf{c},d](\xi)$ (which is called a Farkas' linear assertion)
for Farkas' Lemma by
\[
\Phi[H,\mathbf{c},d](\xi):=\left(\xi\ge \mathbf{0}\right)\wedge \left(\mathbf{A}^\mathrm{T}\xi=\mathbf{c}\right)\wedge\left(\mathbf{b}^{\mathrm{T}}\xi\le d\right)
\]
where $\xi$ is a variable representing a column vector of dimension $m$.

Below, we fix an input probabilistic while loop $P$ with a linear invariant $I$. %
We assume that $P$ is \emph{affine}, i.e.~(i)~every assignment statement in $P$ has an affine expression at its right hand side; and (ii)~the loop guards
of the conditional branches of $P$ are in disjunctive normal form and each atomic proposition is a comparison between affine expressions.
\paragraph{The Synthesis Algorithm for DSM-maps.} Our algorithm for synthesizing DSM-maps consists of the following four steps:
\begin{compactenum}
  \item {\em Template.}
The algorithm establishes a template $\eta$ for a DSM-map by setting
$\eta(\loc,\mathbf{x}):=(\mathbf{\alpha}_\loc)^\mathrm{T}\mathbf{x}+\beta_\loc$ for each $\loc\in \locs$ and
$\mathbf{x}\in\Zset^{|\pvars|}$, where $\mathbf{\alpha}^\loc$ is a vector of scalar variables and $\beta^\loc$ is a scalar variable, both representing unknown coefficients.

\item {\em Variables for Parameters in a DSM-map.}
The algorithm sets up a scalar variable $\epsilon$, two scalar variables $a,b$ and a scalar variable $c$.
These variables directly correspond to the parameters for a DSM-map (see Definition~\ref{def:dbprsm}).

\item {\em Farkas' Linear Assertions.}
From the template, we establish Farkas' linear assertions from the conditions (D1)--(D4).
For example, the condition (D1) at a label $\ell$ requires that for the template $\eta$, it holds that
$\sum_{\mu\in \val{\rvars}} \sampdpd(\mu)\cdot \eta(\loc',u(\nu,\mu))\leq \eta (\loc,\nu)-\epsilon$ for all $\nu\in I(\ell)$.
Since the template $\eta$ is linear and we have affine assignments,
the inequality $\sum_{\mu\in \val{\rvars}} \sampdpd(\mu)\cdot \eta(\loc',u(\nu,\mu))\leq \eta (\loc,\nu)-\epsilon$
would also be linear.
Then (D1) is essentially an inclusion of the set $I(\ell)$ in the halfspace represented by $\sum_{\mu\in \val{\rvars}} \bar{\Upsilon}(\mu)\cdot \eta(\loc',u(\nu,\mu))\leq \eta (\loc,\nu)-\epsilon$, and can be equivalently transformed into a group of Farkas' linear assertions, given that $I(\ell)$ is a finite union of polyhedra.

\item {\em Solution through Linear Programming.}
We group the constructed Farkas' linear assertions together in a conjunctive manner so that we have a system of linear inequalities over scalar variables (including template variables, parameter variables and fresh variables from Farkas' linear assertions). Then, we solve for the variables through linear programming. If we can get a solution for the scalar variables, then we get a DSM-map that witnesses the a.s. termination of the input program; otherwise, the algorithm cannot prove the a.s. termination property and outputs ``\emph{fail}''.
\end{compactenum}

\begin{theorem}
   Linear DSM-maps can be computed in polynomial time.
\end{theorem}
\begin{proof}
	It is straightforward to check that Steps (1)--(3) of the Synthesis algorithm have polynomial runtime. Hence, the resulting LP, which should be solved in Step (4), has polynomial size. It is well-known that LPs can be solved in polynomial time.
\end{proof}

\begin{example}
We now illustrate our synthesis algorithm on the program in Example~\ref{example:rsp}.
\begin{compactitem}
\item First, we set the
template function $\eta(\loc,\mathbf{x})=(\alpha_\loc)^\mathrm{T}\mathbf{x}+\beta_\loc$ for every label $\loc$, where $\mathbf{x}=(x,y)^\mathrm{T}$ is the vector of program variables and the scalar variable $\beta_\loc$ together with the coordinate variables in the vector $\alpha_\loc$  are unknown coefficients at a label $\loc$.
\item
Second, we set up the parameters $\epsilon,a,b,c\in \Rset$ as in the definition of DSM-maps. The unknown coefficients in $\alpha^\loc,\beta^\loc$ and the parameters are what we want to solve for, in order to obtain a concrete DSM-map.
\item In the third step, we establish Farkas' linear assertions.
Below we illustrate an example on the construction of Farkas' linear assertions.
Consider the condition (D4) at  the label $5$. The linear invariant at the label $5$ is  $x\geq -7\wedge 1\leq y\leq 9$ that represents the polyhedron $H= \left\{ \mathbf{x} \mid   \bigl( \begin{smallmatrix} -1 & 0 \\ 0 & 1 \\0&-1 \end{smallmatrix} \bigr)\mathbf{x}\leq \bigl( \begin{smallmatrix} 7 \\ 9\\-1  \end{smallmatrix}  \bigr)\right\}$.
To satisfy (D4), we have to ensure that the following conditions hold for every $\mathbf{x}$: $a\leq\eta(6,\mathbf{x})-\eta(5,\mathbf{x})\leq b$,
$a\leq\eta(7,\mathbf{x})-\eta(5,\mathbf{x})\leq b$ and $\frac{6}{13}\eta(6,\mathbf{x})+\frac{7}{13}\eta(7,\mathbf{x})\leq \eta(5,\mathbf{x})-\epsilon$.
We first rewrite them into
$\bigl(\alpha_5-\alpha_6\bigr)^\mathrm{T}
 \mathbf{x}
 \leq -a+\beta_6-\beta_5 $
 ,
 $(-\alpha_5+\alpha_6)^\mathrm{T}
\mathbf{x}
\leq
b-\beta_6+\beta_5$
and
$(\frac{6}{13}\alpha_6+\frac{7}{13}\alpha_7-\alpha_5)^\mathrm{T}
\mathbf{x}
\leq
-\epsilon$
.
Let $d:=-a+\beta_6-\beta_5$, $d':=b-\beta_6+\beta_5$.
Then we construct the Farkas' linear assertions $\Phi[H,\alpha_5-\alpha_6,d](\xi)$, $\Phi[H,-\alpha_5+\alpha_6,d'](\xi')$ and $\Phi[H,\frac{6}{13}\alpha_6+\frac{7}{13}\alpha_7-\alpha_5,\epsilon](\xi'')$.
\item Finally, in the fourth step we group all generated Farkas' linear assertions together in a conjunctive manner and solve for the unknown coefficients, together with the parameters and the fresh variables from Farkas' linear assertions, using an LP-solver. If we can get a solution for the unknown coefficients, then the algorithm confirms that the input program is a.s. terminating (Theorem~\ref{thm:rsm}). Otherwise, the algorithm outputs ``\emph{fail}''. In this case, our algorithm is able to synthesize a linear DSM-map (see Section~\ref{sec:exp}).
\end{compactitem}

\end{example}

\section{Experimental Results}\label{sec:exp}

In this section, we present experimental results obtained by an implementation of the algorithm of Section~\ref{sec:alg}. Note that our algorithm has very few dependencies, all of which are standard operations (e.g. linear invariant generation and linear programming).

\paragraph{Experimental Benchmarks.} We consider two families of benchmarks:

\begin{compactitem}
	\item First, to illustrate the applicability of our approach to different types of while loops, we consider the program of Figure~\ref{example:counter} (i.e.~the counterexample to the FHV-rule), the Mini-roulette program of Example~\ref{example:rsp}, and three other classical examples of probabilistic programs that exhibit various types of nested while loops (Figure~\ref{fig:benchmark}).
\emph{Program 1} is a simple nested while loop, in which the outer loop control variable is updated in the inner loop.
\emph{Program 2} is a nested while loop with two sequentially-composed inner loops, in which the outer loop control variables are each updated in one of these inner loops.
\emph{Program 3} is a three-level nested while loop.

\item Second, we demonstrate that our approach can handle real-world programs by providing experimental results on the benchmarks used in~\cite{pldi18}.

\end{compactitem}

\begin{figure*}
	\begin{center}
		\begin{tabular}{|c|c|c|}
			\hline
			Program 1 & Program 2 & Program 3\\
			\hline
			\lstset{language=prog}
			\lstset{tabsize=4}
			\begin{lstlisting}[mathescape,basicstyle=\small]
$1:$while $x\geq 1$ do
$2:$	$z:=y$;
$3:$	while $z\geq0$ do
$4:$		$z:=z-1$;
$5:$		$x:=x+r$
$\,\,$	od;
$6:$	$y:=2\times y$;
$7:$	$x:=x-1$
$\,\,$  od
$8:$
			\end{lstlisting}
			&
			\lstset{language=prog}
			\lstset{tabsize=4}
			\begin{lstlisting}[mathescape,basicstyle=\small]		
1:while $x+y\geq 1$ do
2:    $a:=z$;
3:    $b:=z$;
4:   while $a\geq0$ do
5:       $a:=a-1$;
6:       $x:=x+r_1$
     od;
7:   while $b\geq0$ do
8:       $b:=b-1$;
9:       $y:=y+r_2$
     od;
10:    $z:=2\times z$;
11:    $x:=x-1$
   od
12:
			\end{lstlisting}
			&
			\lstset{language=prog}
			\lstset{tabsize=4}
			\begin{lstlisting}[mathescape,basicstyle=\small]		
1: while $x\geq 1$ do
2:    $a:=z$;
3:    while $a\geq0$ do
4:       $a:=a-1$;
5:       $b:=z$;
6:       while $b\geq0$ do
7:           $b:=b-1$;
8:           $x:=x+r$
od;
9:       $x:=x+r$
od;
10:   $z:=2\times z$;
11:   $x:=x-1$
   od
12:
		\end{lstlisting}
			\\ \hline
		\end{tabular}
	\end{center}
	\caption{Our benchmark programs. These programs exhibit different types of nested while loops.}
	\label{fig:benchmark}
\end{figure*}

\paragraph{Invariants.} Our approach is able to synthesize DSMs using very simple invariants obtained from the loop guards. See Appendix~\ref{app:exp} for more details. Note that in all cases, the invariants we use are strictly weaker than, and can be replaced by, invariants generated by standard tools such as~\cite{DBLP:conf/cav/ColonSS03} and~\cite{sting1}. However, we use weaker invariants to demonstrate the power of our algorithm.

\paragraph{Distributions.} We assume that each sampling variable $r$ in Programs 1, 2 and 3 is sampled according to the distribution $\probm(r=1)=0.25, \probm(r=-1)=0.75$. This choice is arbitrary and our approach can synthesize linear DSMs for any distribution, as long as such a DSM exists. The benchmarks of~\cite{pldi18} contain a specification of the distributions.

\paragraph{Implementation and Experiment Machine.} We implemented our approach in Java. Our implementation passed the OOPSLA artifact evaluation and is publicly available at \href{http://pub.ist.ac.at/~akafshda/DSM/}{http://pub.ist.ac.at/$\sim$akafshda/DSM/}.
We used lpsolve~\cite{berkelaar2004lpsolve} and JavaILP~\cite{javailp} for solving the linear programming instances. The results were obtained on a Windows 10 machine with a 2.5 GHz Intel Core i5-2520M processor and 8 GB of RAM.

\paragraph{Experimental Results.}
Table~\ref{tab:experiments} summarizes our experimental results over the five example programs and Table~\ref{tab:experiments2} provides the results over the benchmarks from~\cite{pldi18}. Note that the counterexample program does not terminate almost surely. Therefore, any sound approach is expected to fail on this program.
In all other cases, our approach is extremely efficient. \amir{It processes each benchmark program in less than 2 seconds and successfully synthesizes \emph{separate} linear DSM-maps for each of the while loops in the program}. In all cases, the DSM parameters $\epsilon$ and $c$ are synthesized as $1$ and $0$, respectively (except that $c=-71$ for~\texttt{coupon}). \amir{The reported runtime for each benchmark is the total time spent for synthesizing DSM-maps for all while loops in the benchmark. The column $\eta(\lin)$ reports the expression at the first label $\lin$ of the program in the DSM-map $\eta$ corresponding to the outermost loop.}
\mingzhang{
See Appendix~\ref{app:dsmres} for more details.
}

\begin{table*}
	\begin{center}
		{
			\footnotesize{
				\begin{tabular}{|c|c|c|c|c|c|c|}
					\hline
					Example & Result & Runtime (s) & $\eta(\lin)$ & $[a, b]$\\
					\hline
					\hline
					
					Counterexample & Failure & $0.774$ & -- & --\\
					\hline

					Example~\ref{example:rsp} & Success & $0.812$ & $75.4 \cdot x$ & $[-70.6, 155.6]$ \\
					\hline
					
					Program 1 & Success & $0.722$ & $6 \cdot x + 5$ & $[-4, 8]$\\
					\hline
					
					Program 2 & Success & $0.921$ & $7 \cdot x + 7 \cdot y + 6$ & $[-5, 9.5]$ \\
					\hline
					
					Program 3 & Success & $0.872$ & $8 \cdot x + 7$ & $[-5, 11]$\\
					\hline
					
				\end{tabular}
			}
		}
	\end{center}
	\caption{Experimental Results on Example Programs}
	\vspace{-1em}
	\label{tab:experiments}
\end{table*}

\begin{table*}
	\begin{center}
		{
			\footnotesize{
				\begin{tabular}{|c|c|c|c|c|c|c|}
					\hline
					Benchmark & Result & Runtime (s) & $\eta(\lin)$ & $[a, b]$\\
					\hline
					\hline
					
					\texttt{ber} & Success & $0.597$ & $4 \cdot n - 4 \cdot x + 1$ & $[-3, 1]$ \\
					\hline
					
					\texttt{bin} & Success & $0.822$ & $0.4 \cdot n - 0.4 \cdot x + 1$ & $[-3, 1]$ \\
					\hline
					
					\texttt{C4B\_t09} & Success & $0.665$ & $-4 \cdot j + 4 \cdot x + 1$ & $[-1, 0]$ \\
					\hline
					
					\texttt{C4B\_t13} $\diamondsuit$ & Success & $1.308$ & $4.5 \cdot x + 2 \cdot y - 1$ & $[-1.5, 0.5]$ \\
					\hline
					
					\texttt{C4B\_t15} $\diamondsuit$ & Success & $1.264$  & $5 \cdot x + 2$ & $[-5, 0]$ \\
					\hline
					
					\texttt{C4B\_t19} & Success & $1.202$ & $6 \cdot i + 5$ & $[-4, 2]$ \\
					\hline
					
					\texttt{C4B\_t30} & Success & $0.653$ & $2.5 \cdot x + 2.5 \cdot y + 1$ & $[-3.6, 1.5]$ \\
					\hline
					
					\texttt{C4B\_t61} & Success & $1.217$ & $0.286 \cdot l$ & $[-1.3, 0]$\\
					\hline
					
					\texttt{complex} $\diamondsuit$ & Success & $1.444$ & $8 \cdot N - 8 \cdot x + 1$ & $[-5, 3]$\\
					\hline
					
					\texttt{condand} & Success & $0.637$ & $3 \cdot m + 3 \cdot n + 1$ & $[-1, 0]$ \\
					\hline
					
					\texttt{cooling} $\diamondsuit$ & Success & $1.364$ & $2 \cdot mt - 2 \cdot st + 2.443 \cdot pt + 3.821 $ & $[-2.6, 0]$\\
					\hline
					
					\texttt{coupon} & Success & $0.769$ & $-35 \cdot i + 69$ & \makecell{$[-29, 27]$\\ $c=-71$}\\
					\hline
					
					\texttt{cowboy\_duel} & Success & $0.632$ & $4.066 \cdot \textsf{flag} - 1.162$ & $[-2.7, 2.2]$ \\
					\hline
					
					\texttt{filling\_vol} & Success & $0.720$ & $-3.356 \cdot \textsf{volMeasured} + 3.356 \cdot \textsf{volToFill} + 3$ & $[-16.6, 2.3]$ \\
					\hline
					
					\texttt{geo} & Success & $0.628$ & $7 \cdot \textsf{flag} + 1$ & $[-4, 2]$ \\
					\hline
					
					\texttt{hyper} & Success & $0.599$ & $10 \cdot n - 10 \cdot x$ & $[-19, 1]$ \\
					\hline
					
					\texttt{linear01} & Success & $0.614$ & $1.796 \cdot x + 2.593$ & $[-1.6, 0.2]$ \\
					\hline
					
					\texttt{prdwalk} & Success & $0.661$ & $1.770 \cdot n - 1.770 \cdot x + 1$ & $[-5.6, 3.2]$ \\
					\hline
					
					\texttt{prnes} $\diamondsuit$ & Success & $1.328$ & $-23.655 \cdot n + 0.032 \cdot y + 4.365$ & $[-17.3, 20.3]$ \\
					\hline
					
					\texttt{prseq} & Success & $1.242$ & $1.005 \cdot x - 1.005 \cdot y + 1$ & $[-2, 0]$ \\
					\hline
					
					\texttt{prspeed} & Success & $0.947$ & $8 \cdot m + 4.161 \cdot n - 4.161 \cdot x - 8 \cdot y + 7.484$ & $[-5, 3]$\\
					\hline
					
					\texttt{race} & Success & $0.687$ & $-3.279 \cdot h + 3.279 \cdot t + 14.557$ & $[-17.2, 15.6]$ \\
					\hline
					
					\texttt{rdseql} $\diamondsuit$ & Success & $1.261$ & $5.5 \cdot x + 2 \cdot y - 2$ & $[-1.5, 0.5]$ \\
					\hline
					
					\texttt{rdspeed} & Success & $0.760$ & $6 \cdot m + 2 \cdot n - 2 \cdot x - 6 \cdot y + 2$ & $[-4, 2]$ \\
					\hline
					
					\texttt{rdwalk} & Success & $0.625$ & $6 \cdot n - 6 \cdot x + 1$ & $[-10, 8]$ \\
					\hline
					
					\texttt{rfind\_lv} & Success & $0.625$ & $6 \cdot \textsf{flag} + 1$ & $[-4, 2]$ \\
					\hline
					
					\texttt{rfind\_mc} & Success & $0.668$ & $0.5 \cdot \textsf{flag} - 3.75 \cdot i + 3.75 \cdot k + 1.25$ & $[-1.25, 0]$ \\
					\hline
					
					\texttt{robot} & Success & $1.562$ & $28.778 \cdot N + 114.444$ & $[-216.7, 1.7]$ \\
					\hline
					
					\texttt{roulette} & Success & $0.905$ & $27.595 \cdot \textsf{money} + 205.962$ & $[-205, 109.4]$\\
					\hline
					
					\texttt{sprdwalk} & Success & $0.693$ & $4 \cdot n - 4 \cdot x + 1$ & $[-3, 1]$ \\
					\hline
					
					\texttt{trapped\_miner} $\diamondsuit$ & Success & $1.676$ & $-9 \cdot i + 9 \cdot n + 8$ & $[-6, 4]$\\
					\hline
					
				\end{tabular}
			}
		}
	\end{center}
	\caption{Experimental Results on the Benchmarks of~\cite{pldi18}. Benchmarks that contain nested loops are marked with a $\diamondsuit$.}
	\label{tab:experiments2}
\end{table*}

%

%

%
%
%
%

%
%

\section{Related Works}\label{sect:realtedwork}

We compare our results with the most related works on termination verification of probabilistic programs.
We discuss two main classes of approaches: supermartingale-based and proof-rule-based.

\paragraph{Supermartingale-based approaches.}
The most \amir{related supermartingale-based works %
are}~\cite{vmcai19,MM04,MM05,SriramCAV,HolgerPOPL,BG05,ChatterjeeFG16,ChatterjeeFNH16,ChatterjeeNZ2017,DBLP:journals/pacmpl/McIverMKK18}.
Compared to
these results, the most significant difference is that our result considers modular verification of the termination property, while previous approaches tackle the termination problem directly on the whole program \amir{(except the cases mentioned below)}. In detail, we synthesize individual DSMs for each loop in the program, while most previous results synthesize a global (ranking) supermartingale for the whole program and do not have the modular feature.

\hongfei{\amir{Another} advantage of \amir{our approach} is that we do not require non-negativity of supermartingales, which is however required in all of the previous results mentioned above}. \hongfei{For example, consider Program 1 in Figure~\ref{fig:benchmark}. In this example, we have a DSM-map for the outer loop that only involves the program variable $x$ (see Table~\ref{tab:experiments} and Appendix~\ref{app:dsmres}). Observe that (i)~the expected value of $x$ decreases throughout the outer loop, and (ii)~the value of $x$ \amir{is unbounded and can become arbitrarily positive or negative}. Also note that the decrease in $x$ is the main reason that the loop terminates a.s.
\amir{In} previous \amir{approaches, due to the restrictive requirement} of non-negativity, we cannot choose $x$ as \amir{a} (ranking) supermartingale.
Moreover, we cannot choose $\vert x \vert$ either, \amir{given that} the expected value of $\vert x \vert$ increases at $x=0$.
Hence, \amir{it is non-trivial to obtain} a \emph{non-negative} (ranking) supermartingale for this example. \amir{In contrast, our} approach \amir{is more flexible and succinctly} proves the a.s.~termination \amir{property of this program by synthesizing two distinct DSM-maps:} a DSM on $z$ for the inner loop and a DSM on $x$ for the outer loop.

\amir{We now compare our approach to previous} supermartingale-based results \amir{that provide some level of modularity}.

\amir{\paragraph{Comparison with~\cite{HolgerPOPL}.} This is the most similar result. However, we have already} shown that this approach is not sound. We have also presented the minimal required strengthening and proven that our new approach is sound.

\amir{\paragraph{Comparison with \cite{DBLP:journals/pacmpl/McIverMKK18}.}} Although the approach in~\cite{DBLP:journals/pacmpl/McIverMKK18} is also modular and constructs supermartingales loop-by-loop, their approach has the following disadvantages: (i)~their approach is restricted to non-negative supermartingales, \amir{and cannot be used} when a non-negative supermartingale is hard to construct (see our \amir{above point about} Program~1 in Figure~\ref{fig:benchmark}); (ii)~their approach requires to calculate the complete semantics of the loop body, which is \amir{infeasible in general}, while our approach (together with our algorithmic method) only requires to examine the syntax of the loop body.

\amir{\paragraph{Comparison with~\cite{AgrawalC018}.}} Another related result in~\cite{AgrawalC018} considers lexicographic RSMs that are sound for a.s.~termination of probabilistic programs.
While lexicographic RSMs have some flavor of modularity (such as decomposition based on lexicographic order), they
also synthesize a global lexicographic RSM and hence are not modular in the sense of~(\ref{def:compo}). Moreover, their approach also requires the non-negativity of lexicographic RSMs, thus suffers the same problem when constructing non-negative RSMs is difficult.
}

\paragraph{Proof-rule-based approaches.}
 \amir{Another family of approaches} for termination analysis \amir{are} based on the notion of proof rules~\cite{KaminskiKMO16,DBLP:journals/fac/Hesselink94,JonesPhdThesis,DBLP:conf/lics/OlmedoKKM16,DBLP:journals/pacmpl/McIverMKK18}.
For example,~\cite{KaminskiKMO16} presents a proof-rule-based approach for proving finite expected termination time of
probabilistic while loops, and~\cite{DBLP:conf/lics/OlmedoKKM16} presents sound proof rules for probabilistic programs with recursion.
Most results on proof rules focus on specifying local logical properties at every \amir{label} to ensure a global logical property,
and do not consider \amir{modular} proof rules.
In contrast, we provide \amir{modular} proof rules that prove the almost-sure termination property.
The most relevant result is given in~\cite{pldi18} that presents a \amir{modular} approach for deriving resource bounds of probabilistic programs.
Compared with our result, their result focuses on resource bounds and can only handle programs with finite expected resource consumption, whereas
our result focuses on termination properties and can handle programs with infinite expected termination time.
\hongfei{An example can be obtained by \amir{considering Program 1,  Figure~\ref{fig:benchmark} and} changing the assignment $z:=z-1$ at label 4 to $z:=z+r'$, where we have $\probm(r'=1)=\probm(r'=-1)=0.5$. Then the inner loop models a symmetric walk that terminates a.s. but with infinite expected termination time. \amir{Therefore, this program} has infinite expected termination time.
For this modified example, the original DSM-map  remains \amir{valid} (see ``Program 1'' in Table~\ref{tab:experiments} and Appendix~\ref{app:dsmres}). And thus, \amir{our modular approach proves its} a.s.~termination.
\amir{Note that} our approach only relies on a side condition (existence of a DSM) and the assumption that the loop body is a.s.~terminating, thus \amir{it} can handle loop bodies with infinite expected termination time.}

\section{Conclusion}\label{sect:conclusion}

In this paper, we first proved that a natural probabilistic extension of the variant rule in the Floyd-Hoare logic is not sound for
\amir{modular} verification of almost-sure termination of probabilistic programs and identified the flaw in the previous related work~\cite{HolgerPOPL}.
Then, we proposed a \hongfei{minimal} sound strengthening of the approach in~\cite{HolgerPOPL} \amir{through the notion of descent supermartingales (DSMs)}, and demonstrated an \amir{efficient} algorithmic implementation of our strengthened approach \amir{for linear DSMs}.
An important future direction is to investigate different rules and sound approaches for \amir{modular} verification of probabilistic termination. \amir{Another direction is to consider the algorithmic problem of synthesizing non-linear DSM-maps}.

\subsection*{Acknowledgements}
\begin{small}
The research was partially supported by the National Natural Science Foundation of China (Grant No. 61772336, 61802254), Open Project
of Shanghai Key Laboratory of Trustworthy Computing, %
Vienna Science and Technology Fund (WWTF) Project ICT15-003, Austrian Science Fund (FWF) NFN Grant No. S11407-N23 (RiSE/SHiNE), ERC Starting Grant (279307: Graph Games), an IBM PhD Fellowship, and DOC Fellowship No.~24956 of
the Austrian Academy of Sciences (\"{O}AW).
We are sincerely
grateful to Prof. Yuxi Fu, director of
the support of the BASICS Lab at Shanghai Jiao Tong University, for his support.
\end{small}

\clearpage

%
\newpage
\appendix

\section{The Detailed Semantics}~\label{app:semantics}
The behavior of a probabilistic program $P$ accompanied with its CFG $\mathcal{G}=(\locs,(\pvars,\rvars),\transitions)$ under a scheduler $\sigma$ is described as follows.
The program starts in the initial configuration $(\loc_0,\initval)$.
Then in each \emph{step} $i$ ($i\in\mathbb{N}_0$), given the current configuration $(\loc_{i},\nu_{i})$, the next configuration $(\loc_{i+1},\nu_{i+1})$ is determined by the following procedure:
\begin{compactenum}
\item
a valuation $\mu_i$ of the sampling variables is sampled according to the joint distribution of the cumulative distributions $\{\Upsilon_r\}_{r\in \rvars}$ and independent of all previously-traversed configurations (including $(\loc_i,\nu_i)$), all previous samplings on $\rvars$ and previous executions of probabilistic branches;
\item
if $\loc_i\in\locs_\mathrm{a}$ and $(\loc_i, u, \loc')$ is the only transition in $\transitions$ with source label $\loc_i$, then $(\loc_{i+1},\nu_{i+1})$ is set to be
$(\loc',u(\nu_{i},\mu{r}_i))$.
\item
if $\loc_i\in\locs_\mathrm{b}$ and $(\loc_i, \phi, \loc_1),(\loc_i, \neg\phi,\loc_2)$ are namely the two transitions in $\transitions$ with source label $\loc_i$, then
$(\loc_{i+1},\nu_{i+1})$ is set to be (i) $(\loc_1,\nu_i)$ when $\nu_i\models\phi$ and (ii) $(\loc_2,\nu_i)$ when $\nu_i\models\neg\phi$;

\item
if $\loc_i\in\locs_\mathrm{p}$ and $(\loc_i, p, \loc_1),(\loc_i, 1-p,\loc_2)$ are namely the two transitions in $\transitions$ with source label $\loc_i$, then with a Bernoulli experiment independent of all previous samplings, probabilistic branches and traversed configurations, $(\loc_{i+1},\nu_{i+1})$ is set to be (i) $(\loc_1,\nu_i)$ with probability $p$ and (ii) $(\loc_2,\nu_i)$ with probability $1-p$;
\item
if $\loc_i\in\locs_\mathrm{d}$ and $c_0,\dots,c_i$ is the finite path traversed so far
(i.e., $c_0=(\loc_0,\initval)$ and $c_i=(\loc_i,\nu_i)$) with $\sigma(c_0,\dots,c_i)=(\loc_i, \star, \loc')$, then
$(\loc_{i+1},\nu_{i+1})$ is set to be $(\loc',\nu_i)$;
\item if there is no transition in $\transitions$ emitting from $\loc_i$ (i.e., $\loc_i=\lout$), then $(\loc_{i+1}, \nu_{i+1})$ is set to be $(\loc_i,\nu_{i})$.
\end{compactenum}
We define the semantics of probabilistic programs using Markov decision processes.

\begin{definition}[The Semantics]
The Markov decision process $\MDP{W}=(\MDPstates_W, \MDPactions, \MDPkernel_W)$ (for the probabilistic program $W$) is defined as follows.
\begin{compactitem}
\item The \emph{state space} $\MDPstates_W$ is the configuration set $(\locs\times\val{\pvars})$.
\item The \emph{action set} $\MDPactions$ is $\left\{\tau, \mathbf{th}, \mathbf{el}\right\}$.
Intuitively, $\tau$ refers to absence of nondeterminism and \textbf{th} (resp. \textbf{el}) refers to the \textbf{then}- (resp. \textbf{else}-) branch of a nondeterministic label.
\item The \emph{probability transition function} $\MDPkernel_W:\MDPstates_W\times\MDPstates_W\rightarrow [0,1]$
is given as follows.

For all configurations $(\loc, \nu)$, we have:
\begin{compactitem}
    \item {\em Assignment:}  if $\loc\in\alocs$ is an assignment label and $(\loc,u,\loc')$ is the only triple in $\transitions$ with source label $\loc$ and update function $u$, then
    \[
    \textstyle \MDPkernel_W\left(\left(\loc,\nu\right),\tau ,\left(\ell',\nu'\right)\right):=
    \sum_{\mu\in U}\sampdpd(\mu)\] where $U=\{\mu~|~\ \nu'=u(\nu,\mu)\}$;  \\

    \item {\em Branching:} if $\loc\in\blocs$ and $(\loc, \phi, \loc_1),(\loc, \neg\phi, \loc_2)$ are the two triples in $\transitions$ with source label $\loc$ and propositional arithmetic predicate $\phi$, then
    \[
    \MDPkernel_W\left((\loc,\nu),\tau , (\loc',\nu)\right):=\begin{cases}
    1  & \mbox{ if } \nu\models\phi , \loc'=\loc_1 \\
    1  &  \mbox{ if } \nu\not\models\phi, \loc'=\loc_2 \\
    0 & \mbox{otherwise}
    \end{cases}~;
    \]
   \item {\em Probabilistic:} If $\loc\in\plocs$ and $(\loc, p, \loc_1),(\loc, 1-p, \loc_2)$ are namely two triples in $\transitions$ with source label $\loc$, then
            \[
            \MDPkernel_W\left((\loc,\nu), \tau, (\loc',\nu)\right):=\begin{cases}
            p  & \mbox{ if } \loc'=\loc_1  \\
            1-p  & \mbox{ if } \loc'=\loc_2  \\
            0 & \mbox{ otherwise }\\
            \end{cases}
            \]

    \item {\em Nondeterminism:} If $\loc\in\dlocs$ and $(\loc, \star, \loc_1),(\loc, \star, \loc_2)$ are namely two triples in $\transitions$ with source label $\loc$ such that $\loc_1$ (resp. $\loc_2$) refers to the \textbf{then}-(resp. \textbf{else}-) branch, then
            \[
            \MDPkernel_W\left((\loc,\nu), \mbox{\textbf{th}}, (\loc',\nu)\right):=\begin{cases}
            1  & \mbox{ if } \loc'=\loc_1  \\
            0 & \mbox{ otherwise }\\
            \end{cases}
            \]
            and
             \[
            \MDPkernel_W\left((\loc,\nu), \mbox{\textbf{el}}, (\loc',\nu)\right):=\begin{cases}
            1  & \mbox{ if } \loc'=\loc_2  \\
            0 & \mbox{ otherwise }\\
            \end{cases}
            \]

 \item {\em Terminal label:} if there is no transition in $\transitions$ emitting from $\loc_i$ (i.e., $\loc_i=\lout$), then $\MDPkernel_W\left((\loc,\nu),\tau ,(\loc,\nu)\right):=1$;
    \item for other cases, $\MDPkernel_W\left((\loc,\nu),a, (\loc',\nu')\right):=0$.

\end{compactitem}
\end{compactitem}
\end{definition}

\section{Flaw in the Proof of FHV-rule}\label{app:flaw}
Below we clarify the critical point on where the flaw in~\cite{HolgerPOPL} lies.
The flaw lies in the point that RSMs should be \emph{non-negative}.
In the following, we define an extra technical notion.

\paragraph{Characteristic Random Variables.} Given random variables $X_0,\dots,X_n$
and a predicate $\Phi$,
we denote by $\mathbf{1}_{\phi(X_0,\dots,X_n)}$ the random variable such that
\[
\mathbf{1}_{\phi(X_0,\dots,X_n)}(\omega)=
\begin{cases}
1 & \mbox{if }\phi\left(X_0(\omega),\dots,X_n(\omega)\right)\mbox{ holds} \\
0 & \mbox{otherwise}
\end{cases}
\]
By definition, $\expv\left(\mathbf{1}_{\phi(X_0,\dots,X_n)}\right)=\probm\left(\phi(X_0,\dots,X_n)\right)$.
Note that if $\phi$ does not involve any random variable, then $\mathbf{1}_{\phi}$ can be deemed as a constant whose value depends only on whether $\phi$ holds or not.

This point can also be observed from the counterexample (Figure~\ref{example:counter}) that the value of the program variable $x$ may grow unboundedly below zero due to increasing values of $y$, breaking the non-negativity.
In detail, the flaw lies in their proof of Theorem 7.7
at the claim that the stochastic process satisfying
\[
\condexpv{R_{T_{k+1}}}{\mathcal{F}_{T_k}}\leq R_{T_{k}^{}}-\epsilon\cdot\textbf{1}_{G_0\cap\ldots\cap G_{T_{k}^{}}}
\]
is an RSM.
However, due to the lack of guarantee on the non-negativity of $R_{T_{k+1}}$, we cannot say that this is an RSM, although its conditional expected value decreases in each step.
The rest of their proof tries to remedy this issue by enforcing the stochastic process to be non-negative.
In detail, their proof constructs the stochastic process $R_{T_{k}}\cdot\mathbf{1}_{R_{T_{k}}>0}$ which is non-negative, but then this process may not satisfy the decreasing condition of RSMs.
Thus, no valid RSMs are constructed in their proof, implying that the proof is invalid.

\section{Proof of Theorem~\ref{thm:rsm}}\label{app:thm:rsm}

\noindent\textbf{Theorem~\ref{thm:rsm}.} Let $P=\textbf{while}(G,P').$ If (i)~$P'$ terminates a.s.~for any initial valuation and scheduler; and (ii)~there exists a DSM-map $\eta$ for $P$, then for any initial valuation $\nu^{*}\in\val{\pvars}$ and for all schedulers $\sigma$,
we have $\probm_{\nu^{*}}^\sigma(T< \infty)=1$.
\begin{proof}
   Let $\eta$ be any DSM-map for a program $P$, $\nu_0\in\val{\pvars}$ be any initial valuation and $a,b,c,\epsilon$ be the parameters in Definition~\ref{def:dbprsm}.

   We define the stochastic process $\{X_n=\eta(\loc_n,\nu_n)\}_{n\in \Nset_0}$ adapted to $\{\mathcal{F}_n\}_{n\in\Nset_0}$ representing the evaluation of $P$ according to the semantics.
   If $P$ evaluates to a label $\loc$ with no out transition, then $\eta(\loc,\nu)$ is a constant $c$ by definition.

   Informally, $X_n$ is a ranking supermartingale. If $\{X_n\}_{n\in\Nset_0}$ decreases for sufficiently many times, it will be less than $c$ at $\lin$ which implies termination.
    We have $ X_{B_n}\geq c$ for every $n\in\Nset_0$, where $B_n$ is the stochastic process representing the number of steps of $P$'s $n$-th arrival to the label $\lin$. We suppose that the program $Q$ is terminating for any initial valuation, and thus we have $B_n$ is well defined.
     \begin{eqnarray*}
     \probm(X_{B_n}\geq c)
    &=&\sum_{k=n}^{+\infty}\probm(X_k\geq c\wedge B_n=k)\\
   &\leq & \sum_{k=n}^{+\infty}\probm(X_k\geq c)%
   \end{eqnarray*}
   Let $Y_n=X_n+n\cdot\epsilon$, then $a+\epsilon\leq Y_{n+1}-Y_{n}=X_{n+1}-X_{n}+\epsilon\leq b+\epsilon$.
    \begin{eqnarray*}
      \condexpv{Y_{n+1}}{\mathcal{F}_n}&=& \condexpv{X_{n+1}}{\mathcal{F}_n}+(n+1)\cdot\epsilon\\
      &=& \textbf{1}_{(\loc_n,u,\loc')\in\transitions}\cdot\sum_{\mu\in \val{\rvars}} \bar{\Upsilon}(\mu)\cdot \eta(\loc',u(\nu,\mu)) \\
      &&+\textbf{1}_{(\loc_n,\phi,\loc')\in\transitions\wedge \nu_n\models\phi}\cdot \eta(\loc',\nu_n) \\
         &&+\textbf{1}_{(\loc_n,\star,\loc')\in\transitions}\cdot \eta(\loc',\nu_n) \\
       &&+\textbf{1}_{(\loc_n,p,\loc'), (\loc_n,1-p,\loc'')\in\transitions}\cdot \\&&\quad(p\eta(\loc',\nu_n)+(1-p)\eta(\loc'',\nu_n))\\
       &&+(n+1)\cdot\epsilon \\
      &\leq& \eta(\loc_n,\nu_n)-\epsilon+(n+1)\cdot\epsilon\\
      &=& X_n+n\cdot\epsilon\\
      &=&Y_n
    \end{eqnarray*}
     Thus $\{Y_n\}_{n\in\Nset_0}$ is a supermartingale satisfying the condition of Hoeffding inequality and we have
    \begin{eqnarray*}
   \sum_{k=n}^{+\infty}\probm(X_k\geq c)&=&\sum_{k=n}^{+\infty}\probm(Y_k-Y_0\geq c-X_0+k\cdot\epsilon)\\
   &\leq& \sum_{k=n}^{+\infty}e^{-\frac{2(c-X_0+k\cdot\epsilon)^2}{k(b-a)^2}}\\
   &\leq& \sum_{k=n}^{+\infty}e^{-\frac{2\epsilon^2}{(b-a)^2}k-\frac{4(c-X_0)\epsilon}{(b-a)^2}}
   \end{eqnarray*}
   The above term $\rightarrow 0$ when $n\rightarrow +\infty$,
   And we have \[\probm(T_P< \infty)\geq 1-\lim_{n\rightarrow+\infty}\probm(X_{B_n}\geq c)=1 \]
\end{proof}

\section{An Example Usage of the Proof System $\mathcal{D}$}\label{app:psDSM}

We provide a complete example of applying the proof system $\mathcal{D}$ for proving almost-sure termination of a probabilistic program. 

\begin{example}
		\begin{figure}[H]
		\centering
		\begin{minipage}{5cm}
			\lstset{language=prog}
			\lstset{tabsize=4}
			\begin{lstlisting}[mathescape,basicstyle=\small]
$1:$while $x\geq 0$ do
$2:$	$y := z$;
$3:$	while $y\geq0$ do
$4:$		$y:=y+r$;
$5:$		$x:=x+r$
$\,\,$	od;
$6:$	$x:=x+r$;
$7:$	$z:=2*z$
$\,\,$  od
$8:$
			\end{lstlisting}
		\end{minipage}
		\caption{An Example Program. In this program we have $\probm(r = -1) = 0.75$ and $\probm(r= 1) = 0.25.$}
		\label{fig:dex}
	\end{figure}
	Consider the program in Figure~\ref{fig:dex}, in which $\probm(r = -1) = 0.75$ and $\probm(r= 1) = 0.25$ and therefore, $\expv[r]=-0.5.$ Here is a step-by-step termination proof using the system $\mathcal{D}$\footnote{In all  steps of this proof, we have $\epsilon = 1, c = 0$ and $[a, b] = [-100, 100].$ Note that we do not need to fix/propagate a single $\epsilon$, given that we can simply use the minimum of the $\epsilon$'s used for each individual step for constructing the DSM-maps. The same point applies to $a, b, c$, too.}:
	
	\footnotesize{
	\begin{enumerate}[\textsf{\roman{enumi}.}]
		\vspace{5mm}
		\item $\inference{-100 \leq \expv[R_2 [y\leftarrow y+r]] - R_1 \leq -1}{ \langle R_1\rangle y := y + r \langle R_2\rangle }[(3)] \quad \quad R_1 = 6 \cdot y, \quad R_2 = 6 \cdot y + 2$
		
		\vspace{5mm}
		\item $\inference{-100 \leq \expv[R_3 [x\leftarrow x+r]] - R_2 \leq -1}{ \langle R_2\rangle x := x + r \langle R_3\rangle }[(3)] \quad \quad R_3 = 6 \cdot y + 1$
		
		\vspace{5mm}
		\item $\inference{ \langle R_1\rangle y := y+r \langle R_2\rangle~(\textsf{i}) \\ \langle R_2\rangle x := x+r \langle R_3\rangle~(\textsf{ii}) }{ \langle R_1\rangle y := y + r; x:= x+r \langle R_3\rangle }[(4)]$
		
		\vspace{5mm}
		\item $\inference{ \langle R_1\rangle y := y + r; x:= x+r \langle R_3\rangle~(\textsf{iii})\\ y \geq 0 \rightarrow R_3 \geq 0 \\  y \geq 0 \rightarrow -100 \leq R_1 - R_3 \leq -1 \\ y<0 \rightarrow -100 \leq R_4 - R_3 \leq -1}
		{ \{R_3\} \textbf{ while } y \geq 0 \textbf{ do } y := y+r; x:=x+r \textbf{ od } \{R_4\} }[(1)] \quad \quad R_4 = 6 \cdot y$
		
		\vspace{5mm}
		\item $\inference{}{\tm(y:=y+r)}[(9)]$
		
		\vspace{5mm}
		\item $\inference{}{\tm(x:=x+r)}[(9)]$
		
		\vspace{5mm}
		\item $\inference{\tm(y:=y+r)~(\textsf{v})\\ \tm(x:=x+r)~(\textsf{vi})}{\tm(y:=y+r;x:=x+r)}[(10)]$
		
		\vspace{5mm}
		\item $\inference{\tm(y:=y+r;x:=x+r)~(\textsf{vii})\\\{R_3\} \textbf{ while } y \geq 0 \textbf{ do } y := y+r; x:=x+r \textbf{ od } \{R_4\}~(\textsf{iv})}{\tm( \textbf{ while } y \geq 0 \textbf{ do } y := y+r; x:=x+r \textbf{ od } )}[(8)]$
		
		\vspace{5mm}
		\item $\inference{-100 \leq \expv[R_6[y \leftarrow z]] - R_5 \leq -\epsilon}{ \langle R_5 \rangle y := z \langle R_6 \rangle}[(3)] \quad \quad R_5 = 10 \cdot x + 2, \quad R_6 = 10 \cdot x + 1$
		
		\vspace{5mm}
		\item $\inference{-100 \leq \expv[R_8[y \leftarrow y + r]] - R_7 \leq -1}{ \langle R_7 \rangle y := y+r \langle R_8 \rangle}[(3)] \quad \quad R_7 = 10 \cdot x - 2, \quad R_8 = 10 \cdot x - 3$
		
		\vspace{5mm}
		\item $\inference{-100 \leq \expv[R_6[x \leftarrow x+r]] - R_8 \leq -1}{ \langle R_8 \rangle x := x+r \langle R_6 \rangle}[(3)]$
		
		\vspace{5mm}
		\item $\inference{\langle R_7 \rangle y := y+r \langle R_8 \rangle~(\textsf{x})\\ \langle R_8 \rangle x := x+r \langle R_6 \rangle~(\textsf{xi})}
		{\langle R_7 \rangle y := y+r; x:=x+r \langle R_6 \rangle}[(4)]$
		
		\vspace{5mm}
		\item $\inference{ \langle R_7 \rangle y := y+r; x:=x+r \langle R_6 \rangle~(\textsf{xii}) \\ y \geq 0 \rightarrow -100 \leq R_7 - R_6 \leq -1 \\ y <0 \rightarrow -100 \leq R_9 - R_6 \leq -1 }
		{ \langle R_6 \rangle \textbf{ while } y \geq 0 \textbf{ do } y := y+r; x:=x+r \textbf{ od } \langle R_9 \rangle }[(1)] \quad \quad R_9 = 10 \cdot x$
		
		\vspace{5mm}
		\item $\inference{-100 \leq \expv[R_{10}[x \leftarrow x+r]] - R_9 \leq -1}{ \langle R_9 \rangle x := x + r \langle R_{10} \rangle }[(3)] \quad \quad R_{10} = 10 \cdot x + 4$
		
		\vspace{5mm}
		\item $\inference{-100 \leq \expv[R_{11}[z \leftarrow 2 \cdot z]] - R_{10} \leq -1}
		{ \langle R_{10} \rangle z := 2 * z \langle R_{11} \rangle }[(3)] \quad \quad R_{11} = 10 \cdot x + 3$
		
		\vspace{5mm}
		\item $\inference{\langle R_5 \rangle y := z \langle R_6 \rangle~(\textsf{ix})\\ \langle R_6 \rangle \textbf{ while } y \geq 0 \textbf{ do } y := y+r; x:=x+r \textbf{ od } \langle R_9 \rangle~(\textsf{xiii})}
		{\langle R_5 \rangle y := z;  \textbf{ while } y \geq 0 \textbf{ do } y := y+r; x:=x+r \textbf{ od } \langle R_9 \rangle }[(4)]$
		
		\vspace{5mm}
		\item $\inference{\langle R_9 \rangle x := x + r \langle R_{10} \rangle~(\textsf{xiv}) \\ \langle R_{10} \rangle z := 2 * z \langle R_{11} \rangle~(\textsf{xv})}
		{\langle R_9 \rangle x := x + r; z := 2 * z \langle R_{11} \rangle}[(4)]$
		
		\vspace{5mm}
		\item $\inference{\langle R_5 \rangle y := z;  \textbf{ while } y \geq 0 \textbf{ do } y := y+r; x:=x+r \textbf{ od } \langle R_9 \rangle~(\textsf{xvi})\\\langle R_9 \rangle x := x + r; z := 2 * z \langle R_{11} \rangle~(\textsf{xvii})}
		{\langle R_5 \rangle y := z;  \textbf{ while } y \geq 0 \textbf{ do } y := y+r; x:=x+r \textbf{ od};  x := x + r; z := 2 * z \langle R_{11} \rangle}[(4)]$
		
		\vspace{5mm}
		\item $\inference{\langle R_5 \rangle y := z;  \textbf{ while } y \geq 0 \textbf{ do } y := y+r; x:=x+r \textbf{ od};  x := x + r; z := 2 * z \langle R_{11} \rangle~(\textsf{xviii})\\
		x \geq 0 \rightarrow R_{11} \geq 0\\
		x \geq 0 \rightarrow -100 \leq R_5 - R_{11} \leq -1\\
		x < 0 \rightarrow -100 \leq R_{12} - R_{11} \leq -1 
		}
		{\{R_{11} \} \textbf{while } x \geq 0 \textbf{ do }y := z;  \textbf{ while } y \geq 0 \textbf{ do } y := y+r; x:=x+r \textbf{ od};  x := x + r; z := 2 * z \textbf{ od} \{ R_{12}\}}[(1)]$ \\ \\ $R_{12} = 10 \cdot x$
		
		\vspace{5mm}
		\item $\inference{}{\tm(y := z)}[(9)]$
		
		\vspace{5mm}
		\item $\inference{}{\tm(x:=x+r)}[(9)]$
		
		\vspace{5mm}
		\item $\inference{}{\tm(z:=2*z)}[(9)]$
		
		\vspace{5mm}
		\item $\inference{\tm(x:=x+r)~(\textsf{xxi})\\\tm(z:=2*z)~(\textsf{xxii})}
		{\tm(x:=x+r; z:=2*z)}[(10)]$
		
		\vspace{5mm}
		\item $\inference{\tm(y := z)~(\textsf{xx})\\\tm( \textbf{ while } y \geq 0 \textbf{ do } y := y+r; x:=x+r \textbf{ od }(\textsf{viii})}
		{\tm( y:=z; \textbf{ while } y \geq 0 \textbf{ do } y := y+r; x:=x+r \textbf{ od }) }[(10)]$
		
		\vspace{5mm}
		\item $\inference{\tm( y:=z; \textbf{ while } y \geq 0 \textbf{ do } y := y+r; x:=x+r \textbf{ od })~(\textsf{xxiv})\\
		\tm(x:=x+r; z:=2*z)~(\textsf{xxiii})
		}
		{\tm( y:=z; \textbf{ while } y \geq 0 \textbf{ do } y := y+r; x:=x+r \textbf{ od}; x:=x+r; z:=2*z)}[(10)]$
		
		\vspace{5mm}
		\item $\inference{\tm( y:=z; \textbf{ while } y \geq 0 \textbf{ do } y := y+r; x:=x+r \textbf{ od}; x:=x+r; z:=2*z)~(\textsf{xxv}) \\
		\{ R_{11} \} \textbf{while } x \geq 0 \textbf{ do }y := z;  \textbf{ while } y \geq 0 \textbf{ do } y := y+r; x:=x+r \textbf{ od};  x := x + r; z := 2 * z \textbf{ od} \{ R_{12}\}~(\textsf{xix})}
		{\tm(\textbf{while } x \geq 0 \textbf{ do }y := z;  \textbf{ while } y \geq 0 \textbf{ do } y := y+r; x:=x+r \textbf{ od};  x := x + r; z := 2 * z \textbf{ od})}[(8)]$

	\end{enumerate}
}

\normalsize{
\amir{
	\vspace{1em}
	Note that our approach is modular and uses a distinct DSM-map for each while loop. Specifically, it does \emph{not} synthesize a global RSM for the entire program. For example, in the proof above, we use two distinct DSM-maps for proving the termination of the inner loop and the outer loop. Note that the DSM-map used for the inner loop is in terms of $y$ (see~$R_1, \ldots, R_4$), whereas the DSM-map used for proving a.s.~termination of the outer loop is in terms of $x$ (see $R_5, \ldots, R_{12}$).
}
}

\end{example}

\section{Experimental Results}
\subsection{Invariants Used in the Experiments} \label{app:exp}

The following invariants were used for obtaining experimental results over the example programs:

\noindent Counterexample: \\
$I(1):= $true\\
$I(2):= x\geq 1 $\\
$I(3):= x \geq 0 \wedge z \leq y$\\
$I(4):= z\geq 0 \wedge z\leq y$\\
$I(5):= x \leq 1 \wedge z\geq 0\wedge z\leq y$\\
$I(6):= x \geq 2 \wedge z\geq 0\wedge z\leq y$\\
$I(7):= z \geq 0\wedge z\leq y$\\
$I(8):= z\leq -1 \wedge z\leq y$\\
$I(9):= z\leq -1 \wedge z \leq 0.25 \cdot y$\\

\medskip

\noindent Program 1:\\
$I(1):= $true\\
$I(2):= x\geq 1 $\\
$I(3):=  z \leq y$\\
$I(4):= z\geq 0 \wedge z\leq y$\\
$I(5):= z\geq -1 \wedge z\leq y-1$\\
$I(6):= z\leq -1 \wedge z\leq y$\\
$I(7):= z\leq -1$\\

\medskip

\noindent Program 2:\\
$I(1):= $true\\
$I(2):= x+y\geq 1 $\\
$I(3):= x+y\geq 1 \wedge a=z$\\
$I(4):= a \leq z \wedge b=z$\\
$I(5):= a\geq 0 \wedge a\leq z\wedge b=z $\\
$I(6):= a\geq -1 \wedge a\leq z-1\wedge b=z$\\
$I(7):= a\leq -1 \wedge a\leq z\wedge b \leq z$\\
$I(8):= a\leq -1 \wedge a\leq z\wedge b\geq 0 \wedge b\leq z$\\
$I(9):= a\leq -1 \wedge a\leq z\wedge b\geq -1 \wedge b\leq z-1$\\
$I(10):= a\leq -1 \wedge a\leq z\wedge b\leq -1 \wedge b\leq z$\\
$I(11):= a\leq -1 \wedge a \leq 0.5 \cdot z \wedge b\leq -1 \wedge b \leq 0.5 \cdot z$\\

\medskip

\noindent Program 3:\\
$I(1):= $true\\
$I(2):= x\geq 1 $\\
$I(3):= a\leq z$\\
$I(4):= a\geq 0 \wedge a\leq z$\\
$I(5):= a\geq -1 \wedge a\leq z-1 $\\
$I(6):= a\geq -1 \wedge a\leq z-1\wedge b \leq z$\\
$I(7):= a\geq -1 \wedge a\leq z-1\wedge b\geq0\wedge b\leq z$\\
$I(8):= a\geq -1 \wedge a\leq z-1\wedge b\geq -1 \wedge b\leq z-1$\\
$I(9):= a\geq -1 \wedge a\leq z-1\wedge b\leq -1 \wedge b\leq z$\\
$I(10):= a\leq -1 \wedge a\leq z$\\
$I(11):= a\leq -1 \wedge a \leq 0.5 \cdot z$\\

\subsection{Details of the Synthesized DSM-maps} \label{app:dsmres}

Our implementation produced the following DSM-maps for the outer-most loops of the example programs:

\begin{table}[h]
	\begin{center}
		$\epsilon = 1, c = 0, [a, b] = [-70.6, 155.6]$
		
		\begin{tabular}{|l|l|}
			\hline
			$\loc$ & $\eta(\loc, \mathbf{\nu})$\\
			\hline
			\hline
			
			$1$ & $75.4 \cdot x$\\
			\hline
			$2$ & $75.4 \cdot x - 1$\\
			\hline
			$3$ & $75.4 \cdot x - 0.8 \cdot y + 2$ \\
			\hline
			$4$ & $75.4 \cdot x - 0.8 \cdot y + 1$\\
			\hline
			$5$ & $75.4 \cdot x - 0.8 \cdot y$\\
			\hline
			$6$ & $75.4 \cdot x - 0.8 \cdot y + 80.2$\\
			\hline
			$7$ & $75.4 \cdot x - 0.8 \cdot y - 70.6$\\
			\hline
			$8$ & $75.4 \cdot x - 0.8 \cdot y$\\
			\hline
			$9$ & $75.4 \cdot x - 0.8 \cdot y + 155.6$\\
			\hline
			$10$ & $75.4 \cdot x - 0.8 \cdot y - 70.6$\\
			\hline
			$11$ & $75.4 \cdot x - 0.8 \cdot y + 3.8$\\
			\hline
		\end{tabular}
	\end{center}
	\caption{The Synthesized DSM-map for Example~\ref{example:rsp}}
\end{table}

\begin{table}[h]
	\begin{center}
		$\epsilon = 1, c = 0, [a, b] = [-4, 8]$
		
		\begin{tabular}{|l|l|}
			\hline
			$\loc$ & $\eta(\loc, \mathbf{\nu})$\\
			\hline
			\hline
			
			$1$ & $6 \cdot x + 5$\\
			\hline
			$2$ & $6 \cdot x + 4$\\
			\hline
			$3$ & $6 \cdot x + 2$\\
			\hline
			$4$ & $6 \cdot x + 1$\\
			\hline
			$5$ & $6 \cdot x$\\
			\hline
			$6$ & $6 \cdot x + 1$\\
			\hline
			$7$ & $6 \cdot x$\\
			\hline
		\end{tabular}
	\end{center}
	\caption{The Synthesized DSM-map for Program 1}
\end{table}

\begin{table}[h]
	\begin{center}
		$\epsilon = 1, c = 0, [a, b] = [-5, 9.5 ]$
		
		\begin{tabular}{|l|l|}
			\hline
			$\loc$ & $\eta(\loc, \mathbf{\nu})$\\
			\hline
			\hline
			
			$1$ & $7 \cdot x + 7 \cdot y + 6$\\
			\hline
			$2$ & $7 \cdot x + 7 \cdot y + 5$\\
			\hline
			$3$ & $7 \cdot x + 7 \cdot y + 4$\\
			\hline
			$4$ & $7 \cdot x + 7 \cdot y + 3$\\
			\hline
			$5$ & $7 \cdot x + 7 \cdot y + 1.5$\\
			\hline
			$6$ & $7 \cdot x + 7 \cdot y + 0.5$\\
			\hline
			$7$ & $7 \cdot x + 7 \cdot y + 2$\\
			\hline
			$8$ & $7 \cdot x + 7 \cdot y + 1$\\
			\hline
			$9$ & $7 \cdot x + 7 \cdot y $\\
			\hline
			$10$ & $7 \cdot x + 7 \cdot y + 1$ \\
			\hline
			$11$ & $7 \cdot x + 7 \cdot y$\\
			\hline
		\end{tabular}
	\end{center}
	\caption{The Synthesized DSM-map for Program 2}
\end{table}

\begin{table}[h]
	\begin{center}
		$\epsilon = 1, c = 0, [a, b] = [-5, 11]$
		
		\begin{tabular}{|l|l|}
			\hline
			$\loc$ & $\eta(\loc, \mathbf{\nu})$\\
			\hline
			\hline
			
			$1$ & $8 \cdot x + 7$\\
			\hline
			$2$ & $8 \cdot x + 3$\\
			\hline
			$3$ & $8 \cdot x + 2$\\
			\hline
			$4$ & $8 \cdot x + 1$\\
			\hline
			$5$ & $8 \cdot x$\\
			\hline
			$6$ & $-b + 8 \cdot x + z - 1$\\
			\hline
			$7$ & $-b + 8 \cdot x + z - 2$\\
			\hline
			$8$ & $-b + 8 \cdot x + z - 4$\\
			\hline
			$9$ & $8 \cdot x - 1$\\
			\hline
			$10$ & $8 \cdot x + 1$\\
			\hline
			$11$ & $8 \cdot x$\\
			\hline
		\end{tabular}
	\end{center}
	\caption{The Synthesized DSM-map for Program 3}
\end{table}

\end{document}